\documentclass[12pt, journal, draftclsnofoot, onecolumn]{IEEEtran}
\usepackage{lettrine}
\usepackage{geometry}
\usepackage{setspace}
\usepackage{amssymb}
\usepackage{amsthm}
\usepackage{cite}
\usepackage{amsmath}
\usepackage{amssymb}
\usepackage{booktabs}
\usepackage{subfigure}
\usepackage{multicol}
\usepackage{graphicx}
\usepackage{epsfig}
\usepackage{epstopdf}
\usepackage{cases}
\usepackage{verbatim}
\usepackage{algorithmic}
\usepackage[ruled,vlined]{algorithm2e}
\usepackage{bm}
\usepackage{amsmath, amssymb, amsthm}
\usepackage{color}
\usepackage{multirow}
\usepackage{float}

\theoremstyle{definition} 
\theoremstyle{definition} 
\theoremstyle{plain} \newtheorem{lemma}{Lemma}
\theoremstyle{plain} \newtheorem{proposition}{Proposition}
\theoremstyle{plain} \newtheorem{corollary}{Corollary}

\usepackage{mathrsfs}
\usepackage{threeparttable}
\usepackage{ifpdf}
\ifCLASSINFOpdf
\else
\fi

\begin{document}

%
\title{Association and Caching in Relay-Assisted mmWave Networks: From A Stochastic Geometry Perspective}


\author{Zhuojia~Gu, 
        Hancheng~Lu,~\IEEEmembership{Senior Member,~IEEE,}
        Ming~Zhang,
        Haizhou~Sun,
        and Chang Wen Chen,~\IEEEmembership{Fellow,~IEEE}
\IEEEcompsocitemizethanks{
\IEEEcompsocthanksitem This work was supported in part by the National Science Foundation of China (No.61771445, No.91538203) and the Fundamental Research Funds for the Central Universities.
\IEEEcompsocthanksitem Zhuojia Gu, Hancheng Lu, Ming Zhang and Haizhou Sun are with CAS Key Laboratory of Wireless-Optical Communications, University of Science and Technology of China, Hefei 230027, China (email: guzj@mail.ustc.edu.cn; hclu@ustc.edu.cn; mzhang95@mail.ustc.edu.cn; hzs128@mail.ustc.edu.cn).
\IEEEcompsocthanksitem Chang Wen Chen is with The Chinese University of Hong Kong, Shenzhen, Guangdong, 518172, China. He is also with the State University of New York, Buffalo, NY 14260 USA (email: chencw@cuhk.edu.cn; chencw@buffalo.edu).

}
}

\maketitle

%


\vspace{-10mm}
\begin{abstract}

Limited backhaul bandwidth and blockage effects are two main factors limiting the practical deployment of millimeter wave (mmWave) networks. To tackle these issues, we study the feasibility of relaying as well as caching in mmWave networks. \textcolor{blue}{A user association and relaying (UAR) criterion dependent on both caching status and maximum biased received power is proposed by considering the spatial correlation caused by the coexistence of base stations (BSs) and relay nodes (RNs). A joint UAR and caching placement problem is then formulated to maximize the backhaul offloading traffic. Using stochastic geometry tools, we decouple the joint UAR and caching placement problem by analyzing the relationship between UAR probabilities and caching placement probabilities. We then optimize the transformed caching placement problem based on polyblock outer approximation  by exploiting the monotonic property in the general case and utilizing convex optimization in the noise-limited case.} Accordingly, we propose a BS and RN selection algorithm where caching status at BSs and maximum biased received power are jointly considered. \textcolor{blue}{Experimental results demonstrate a significant enhancement of backhaul offloading using the proposed algorithms}, and show that deploying more RNs and increasing cache size in mmWave networks is a more cost-effective alternative than increasing BS density to achieve similar backhaul offloading performance.

\end{abstract}

\vspace{-3mm}
\par

%
\IEEEpeerreviewmaketitle
\vspace{5mm}

\section{Introduction}
With the rapid growth of wireless network traffic, available spectrum in conventional sub-6 GHz networks appears to be incapable of meeting the ever-increasing demand in the near future. According to Cisco's forecast, traffic from wireless and mobile devices will account for 71 percent of total IP traffic by 2022 \cite{cisco}. In this regard, exploring higher radio spectrum to increase wireless network capacity is imperative. The use of millimeter wave (mmWave) frequency bands, between 30 and 300 GHz, has gained tremendous research interest and becomes a good candidate for the fifth generation (5G) cellular networks \cite{Coverage-and-rate-analysis, Stochastic-geometry-modeling, Millimeter-wave-mobile}. To achieve acceptable coverage and rate, BSs need to be densely deployed in mmWave networks to overcome the shortcoming in transmission distance of mmWave signals. This poses a particular challenge for the backhaul network, especially given the extremely high rates resulting from mmWave bandwidths on the order of GHz.
Although the high-speed optical fiber backhaul provides a theoretical solution, in practice, it is rather arduous and costly to connect the core server to all densely deployed BSs with fibers \cite{Caching-policy-toward-maximal}. Alternatively, microwave backhaul is considered more practical, but it may pessimistically become a bottleneck of the network throughput brought by the network densification. To address this challenge, caching popular contents at BSs has been proposed as one of the most effective solutions \cite{A-survey-of-energy-efficient, Approximation-algorithms-for-mobile, Distributed-caching-for-data}, considering the fact that most mobile data are popular contents such as video, maps and news, which are repeatedly requested and are cacheable.

\vspace{3mm}
On the other hand, mmWave communication also suffers from higher pathloss of signals governed by the Friis free-space equation \cite{Wireless-Communications-Principles}, and much more severe blockage effects compared with sub-6 GHz signals. Distinct differences between line-of-sight (LOS) and non-line-of-sight (NLOS) propagation characteristics are among the major problems for ensuring coverage when utilizing mmWave bands.
\textcolor{blue}{As a direct byproduct, the possibility that a cache-hit user can obtain a sufficient downlink data rate from the associated BS is reduced, which greatly impairs the effectiveness of caching in mmWave networks to alleviate backhaul bandwidth.}
To solve this issue, the deployment of Relay Nodes (RNs), as network infrastructure elements without a dedicated wired backhaul connection, has been regarded as a cost-effective way to overcome the blockage effects and increase the coverage probability in mmWave networks \cite{Stochastic-Geometry-Modeling-and-System,Performance-analysis-of-opportunistic}. The combination of caches and relay nodes in mmWave networks is expected to improve the performance of signal coverage as well as backhaul traffic offloading, thus alleviating the backhaul bottleneck and achieving high throughput in mmWave networks.
\textcolor{blue}{Towards this end, this paper studies the backhaul offloading optimization in cache-enabled and relay-assisted mmWave networks.}

\subsection{Related Work}
\textit{1) Caching in mmWave networks}: Caching placement with finite cache size is the key issue in wireless local cache.  Generally, there are two approaches to implement wireless local cache, i.e., the deterministic caching placement and the probabilistic caching placement. The deterministic caching placement \cite{Approximation-algorithms-for-mobile, Distributed-caching-for-data, Joint-optimization-of-caching} typically requires obtaining the information of the network node locations and the instantaneous channel states, which results in the high complexity and fewer tractable performances due to the randomness of the geographic distribution of network nodes and time-varying wireless channels.
\textcolor{blue}{By contrast, the probabilistic caching placement \cite{Caching-Placement-in-Stochastic} permits simple implementation and has a good tractability. Tractable expressions of various network metrics can be obtained by modeling the distributions of network nodes as homogeneous Poisson point processes, and different probabilistic caching strategies were proposed to improve the network performance such as cache hit probability\cite{Probabilistic-caching-in-wireless}, coverage\cite{Analysis-and-optimization-of-caching, An-analysis-on-caching-placement}, throughput\cite{Caching-policy-toward-maximal } or even profits of network service providers\cite{Pricing and resource allocation}.}

\textcolor{blue}{Recently, the benefits of caching in mmWave networks have been demonstrated in \cite{D2D-Aware-Device-Caching, Beef-Up-mmWave-Dense, Caching-Meets-Millimeter-Wave, Proactive-Caching-for-Mobile }. Authors in \cite{D2D-Aware-Device-Caching} designed a caching policy by exploiting the directionality of the mmWave band to achieve higher offloading and lower content-retrieval delays, and authors in \cite{Beef-Up-mmWave-Dense} extended the work by taking practical directional antenna model and cooperative caching into consideration. Caching was leveraged to address the short connection durations and frequent handoffs due to mmWave antenna directionality and high user mobility in \cite{Caching-Meets-Millimeter-Wave}, and \cite{Proactive-Caching-for-Mobile} extended the work to enhancing the quality of mobile video streaming.
In addition, caching in hybrid mmWave and sub-6 GHz networks has been investigated in \cite{Cache-enabled-hetnets, An-analysis-on-caching-placement, Content-placement-in-cache-enabled}, where the key channel features at mmWave and sub-6 GHz frequencies were taken into account. }

\textit{2) Relaying in mmWave networks}:
Recently, it was shown in \cite{Key-elements-to-enable} that the use of relays can also be a promising solution for mmWave networks to combat the blockage effects as well as path loss. To be specific, when a LOS beam blockage occurs and the NLOS links cannot give satisfactory channel quality, the transmitter may steer the beam direction to a nearby relay with a LOS link to the destination.
\textcolor{blue}{Many existing works focused on the optimization of relay-assisted mmWave networks, such as relay placement \cite{Toward-robust-relay-placement}, cell association \cite{Distributed-association-and-relaying}, as well as relay selection and scheduling \cite{Relay-assisted-and-QoS}. }
\textcolor{blue}{Moreover, some researches have analytically investigated the performance of relay-assisted mmWave networks \cite{On-the-Performance-of-Relay, On-the-Performance-of-mmWave-Networks, Improving-the-coverage-and-spectral}. }
Authors in \cite{On-the-Performance-of-Relay} studied the coverage probability and transmission capacity of relay-aided mmWave networks under two different relay selection techniques. Similarly, a coverage probability of an outdoor mmWave ad hoc network aided by relays leveraging tools from stochastic geometry was reported in \cite{On-the-Performance-of-mmWave-Networks}. Authors in \cite{Improving-the-coverage-and-spectral} studied the impact of device-to-device relays on the coverage and spectral efficiency of mmWave networks.

\subsection{Motivation and Contribution}
From the discussions presented above, it is expected that when caches and relays are introduced simultaneously in the mmWave networks, it is possible to improve the performance of signal coverage as well as backhaul traffic offloading, thus alleviating the backhaul bottleneck and achieving high throughput in mmWave networks.
However, the introduction of both relays and caches into mmWave networks poses new challenges in performance analysis and optimization.
\textcolor{blue}{First, user association becomes more flexible when relays are introduced in mmWave networks. To overcome the blockage effects, whether a user should be directly associated to a BS or associated to a BS with the aid of relays is a challenging issue that requires careful design.
Second, caching placement should be reconsidered due to the introduction of relays in mmWave networks. The spatial correlation caused by the coexistence of BSs and RNs complicates the caching placement optimization of backhaul offloading.
More importantly, the user association and relaying issue is coupled with the caching placement issue when considering the backhaul offloading optimization. The user association and relaying strategy will have an impact on the caching placement strategy, and vice versa.}
Therefore, the simultaneous application of caches and relays in mmWave networks poses fundamental new challenges, which shall require further rigorous analysis and optimization. The main contributions of this research can be summarized as below:
\begin{itemize}
  \item  We develop an analytical network model for a cache-enabled and relay-assisted downlink mmWave network based on a stochastic geometry principle.
      \textcolor{blue}{User association and relaying (UAR) are jointly considered with caching placement to offload backhaul traffic  in the considered scenario. We decouple the joint UAR and caching placement problem by using stochastic geometry tools to analyze the relationship between UAR probabilities and caching probabilities.
      Based on the analytical results, the backhaul offloading optimization problem can be formulated as a caching placement optimization problem.}
  \item To measure the backhaul offloading performance, we define the successful backhaul offloading probability (SBOP) as the performance metric.
      \textcolor{blue}{We propose the user association and relaying criterion dependent on both caching status at BSs and maximum biased received power. By taking into account the spatial correlation due to the coexistence of BSs and RNs, as well as the blockage effects in mmWave networks, the distribution of distances between different communicating nodes are derived, and the user association probabilities are obtained. Based on the derived relationship between user association probabilities and caching probabilities,} a holistic analytical expression of SBOP can be obtained. Under the noise-limited scenario, a closed-form expression of the SBOP is also presented.
  \item \textcolor{blue}{By decoupling the UAR problem and caching placement problem through the analytical work, we formulate the caching placement optimization problem with the goal to maximize SBOP.} By exploiting the monotonic property of this optimization problem, we design an efficient algorithm based on polyblock outer approximation to find the globally optimal solution. Furthermore, inspired by the noise-limited assumption in mmWave networks, we also design a suboptimal algorithm utilizing convex optimization to approximate the optimal solution with low complexity. Based on the optimized caching probabilities, a BS/RN selection algorithm is then developed for practical deployment.
  \item Numerical simulations have been carried out and the results demonstrate that the proposed algorithms outperform existing algorithms in terms of SBOP. Furthermore, the results also indicate that deploying more relay nodes and increasing cache size in mmWave networks can be promising alternative approaches to achieve comparable backhaul offloading performance as deploying more mmWave base stations.
\end{itemize}

The rest of this paper is organized as follows. We first introduce the system model in Section \uppercase\expandafter{\romannumeral2}. \textcolor{blue}{We then formulate the joint UAR and caching placement problem and analyze the user association and relaying probabilities to reformulate the problem in Section \uppercase\expandafter{\romannumeral3}.} In Section \uppercase\expandafter{\romannumeral4}, we first describe two algorithms based on the polyblock outer approximation and convex optimization, respectively, to maximize the proposed SBOP, and then present the BS/RN selection algorithm. The performance evaluations and their associated discussions are provided in Section \uppercase\expandafter{\romannumeral5} and Section \uppercase\expandafter{\romannumeral6} concludes this paper with summary.

\section{System Model}
\subsection{Network Model}
We consider downlink transmissions in a cache-enabled and relay-assisted mmWave network as shown in Fig. \ref{fig:1}. The locations of BSs, RNs and UEs are modeled as three independent and homogeneous Poisson Point Processes (PPPs), which are denoted by $\Phi_{\mathrm{BS}}$, $\Phi_{\mathrm{RN}}$, and $\Phi_{\mathrm{UE}}$ with densities $\lambda_{\mathrm{BS}}$, $\lambda_{\mathrm{RN}}$, and $\lambda_{\mathrm{UE}}$, respectively. A LOS link is defined as a link with no blockage between the transceivers, while an NLOS link is defined as a link with blockage between the transceivers. Generally speaking, a link of shorter distance will have a higher probability to be a LOS link. Note that by introducing RNs, the average distance between a UE and its associated node can be shortened, so more LOS links are expected with the assistance of RNs. We assume that the universal frequency reuse is applied among all BSs and RNs, meaning that BSs and RNs share the same carrier frequencies/wavelengths, similar to the scenarios defined in \cite{On-the-Performance-of-mmWave-Networks, Coverage-in-heterogeneous}. A two-slot Decode-and-Forward (DF) relaying strategy is assumed for the BS-RN-UE link \cite{Coverage-analysis-of-decode, On-the-Performance-of-Relay}. In the first time slot, BSs transmit signals to RNs, while in the second time slot, RNs forward the data (decoded from the received signal in the first time slot) to the UEs.


In the relay-assisted mmWave network, we equip BSs with caches. Assuming a finite file category $\mathcal{F} = \{ c_1,\cdots,c_F\}$, we adopt a probabilistic cache policy described in Section \ref{caching-model}. A randomly selected \textit{typical UE}\footnote{Using Slivnyak's theorem, placing a typical point at the origin does not change the statistics of the PPP, so we consider the analysis of a typical UE to evaluate the network performance.} is fixed at the origin, which is denoted by UE$_0$. UE$_0$ is associated with either a BS or a RN according to the association criterion described in Section \ref{association-criterion}. When UE$_0$ requests a file, a cache hit event may happen in two cases:
\begin{itemize}
  \item \textbf{One-hop link: Direct cache hit}, when UE$_0$ directly communicates with the serving BS, which is denoted by BS$_0$, and the requested file is cached in BS$_0$.
  \item \textbf{Two-hop links: Relay-assisted cache hit}, when UE$_0$ communicates with the serving BS$_0$ via the help of the serving RN, denoted by RN$_0$, and the requested file is cached in BS$_{R0}$.
\end{itemize}

In the case where the requested file cannot be found in BS$_0$, the file is retrieved from the core network through the backhaul links, and then transmitted to the UE. Denote the one-hop link between BS$_0$ and UE$_0$ as the BU link, and the two-hop links as the BR link and the RU link, respectively.
\textcolor{blue}{Denote the user association and relaying probabilities for the three types of links as $\chi_{\mathrm{BU}}$, $\chi_{\mathrm{BR}}$, and $\chi_{\mathrm{RU}}$, respectively.}
\newgeometry{left=1in,right=1in,top=0.84in,bottom=0.84in}
\begin{figure}
\centering
\includegraphics [width = 4.0in]{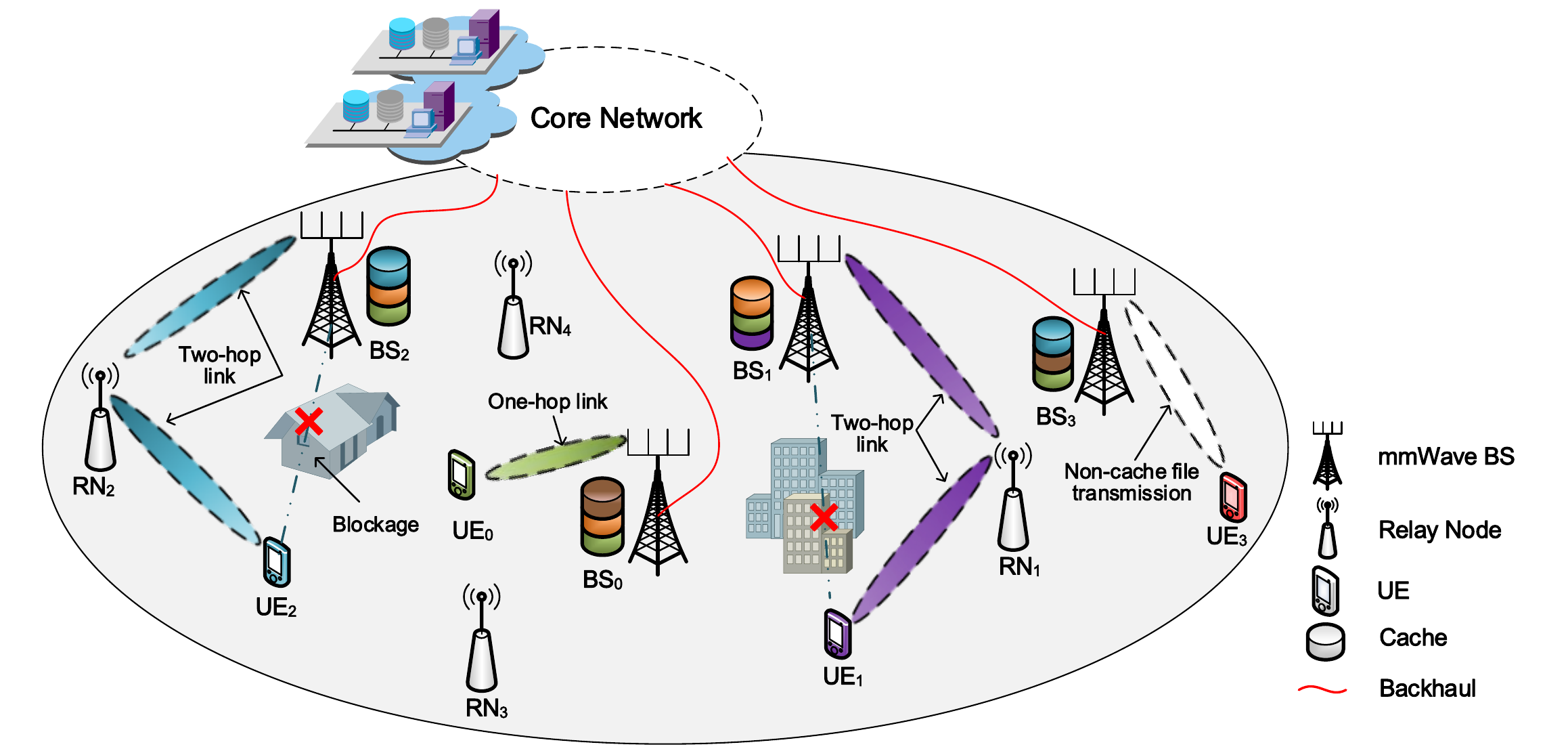}
\caption{An illustration of a relay-assisted millimeter wave network with caches at BSs.}
\label{fig:1}
\end{figure}
\begin{spacing}{1.38}
\subsection{Caching Model}\label{caching-model}
Proactive caching is adopted in this paper, where popular files are pre-cached in BSs when the network is off-peak. Assume that each BS can cache a total of $C$ files of equal size\footnote{For analytical simplicity, the size of each file is assumed to be equal and normalized to one. In practice, if the file sizes are unequal, we can split the files into small partitions of equal size, thus the file sizes can still be regarded as equal.}. Further, we model file requests by a Zipf distribution, which is widely adopted in content-centric networking scenarios \cite{Joint-caching-and-pricing, Delay-and-Power-Tradeoff, Joint-power-allocation-and-caching}, and the popularity of the $n$-th file is given as $a_n = \frac{n^{- \delta}}{\sum_{m=1}^{F}m^{-\delta}}$, where $F$ is the total number of files in the network, and $\delta$ is a shape parameter that shapes the skewness of the popularity distribution. Thus, a file with smaller index owns higher popularity.

We adopt a probabilistic caching algorithm to place these files \cite{Caching-policy-toward-maximal}. Denote $p_{n}$ as the probability of caching the $n$-th file $c_n$ such that $0\leq p_{n} \leq 1 $ and $\sum_{n=1}^{F} p_{n} \leq C$. Each BS caches files independently based on the caching probability distribution $\textbf{p} = \{p_1,\cdots,p_F\}$. According to the thinning theorem of PPP \cite{Stochastic-Geometry-and-Its}, the locations of BSs caching the $n$-th file can be modeled as an independent PPP with density $\lambda_{\mathrm{BS_n}} =  p_{n} \lambda_{\mathrm{BS}}$.

\subsection{Transmission Model}\label{association-criterion}

We adopt a two-state statistical blockage model for each link as in \cite{Coverage-and-rate-analysis}, such that the probability of a link to be LOS or NLOS is a function of the distance between UE$_0$ and the serving BS$_0$ (or RN$_0$). Assume that the distance between them is $r$, then the probability that a link of length $r$ is LOS or NLOS can be modeled as
\vspace{-2mm}
\begin{equation}\label{link-state-def}
  \rho_\mathrm{L}(r) = \mathrm{e}^{-\beta r},\ \rho_\mathrm{N}(r) = 1 - \mathrm{e}^{-\beta r},
\end{equation}
respectively, where $\beta$ is the blockage parameter depending on the building parameter and density \cite{Analysis-of-blockage}. Based on the thinning theorem, the PPP $\Phi_\omega$ can be thinned as $\Phi_\omega^{\mathrm{L}}$ and $\Phi_\omega^{\mathrm{N}}$ with densities $\lambda_\omega \rho_\mathrm{L}(r)$ and $\lambda_\omega \rho_\mathrm{N}(r)$ for LOS and NLOS, respectively, where $\omega \in \{\mathrm{BS,RN}\}$.

It is assumed that the antenna arrays at the BSs, RNs and UEs perform directional beamforming with the main lobe directed towards the dominant propagation path and having less radiant energy in other directions. For tractability in the analysis, we adopt a sectorial antenna pattern \cite{Performance-analysis-of-outdoor}, and the antenna gain pattern at the transceivers is given as

\vspace{-3mm}
\begin{small}
\begin{equation}\label{Gq}
  G_q(\theta)=\begin{cases}
M, \quad \text{if} \ |\phi| \leq \theta,   \\
m, \quad \text{if} \ |\phi| \geq \theta,
\end{cases}
\end{equation}
\end{small}
where $\theta$ is the main lobe beamwidth, $q \in \{\text{T, R}\}$ denotes the transmitter and the receiver, respectively, $\phi \in [0, 2\pi)$ is the angle deviation from the antenna boresight, and $M$ and $m$ are the directivity gain of main and side lobes, respectively. Thus, the random antenna gain/interference $G$ between the transceiver has 3 patterns with different probabilities, which is given as

\vspace{-7mm}
\begin{small}
\begin{numcases}{G=}
MM, \quad \text{with prob.} \  (\frac{\theta}{2\pi} )^2, \nonumber  \\
Mm, \quad \text{with prob.} \  \frac{2 \theta (2\pi - \theta)}{(2\pi)^2}, \label{Gx-definition} \\
mm, \quad \text{with prob.} \  (\frac{2\pi-\theta}{2\pi} )^2. \nonumber
\end{numcases}
\end{small}
$\!\!$For tractability, we assume a perfect beam alignment between the transmitter and its associated receiver \cite{On-the-Performance-of-Relay}.



Independent Nakagami fading is assumed for each link. Parameters of Nakagami fading $N_\mathrm{L}$ and $N_\mathrm{N}$ are assumed for LOS and NLOS links, respectively. Let $h_\ell$ be the small-scale fading term on the $\ell$-th link. Then $|h_\ell|^2$ is a normalized Gamma random variable. Further, we assume $N_\mathrm{L}$ and $N_\mathrm{N}$ are positive integers for simplicity. Accordingly, the channels in the relay-assisted downlink mmWave network are modeled as follows,

\textit{1) One-hop link: } When UE$_0$ is directly associated with BS$_0$ and requests the $n$-th file from BS$_0$ via this one-hop link (i.e., BU link), the downlink received signal-to-interference-plus-noise ratio (SINR) is given by
\begin{equation} \small \label{SINR-def1}
\text{SINR}^{\mathrm{BU}}_{n} = \frac{P_{\mathrm{BS}} |h_n|^2 G_{\mathrm{BS}} L(|r_{\mathrm{BS}_{0}, \mathrm{UE}_{0}}|)}{ \underbrace{ \sum_{j \in \Phi_{\mathrm{BS}_n} \backslash \{BS_0\} } P_{\mathrm{BS}} |h_j|^2 G_j L(|r_j|) + \sum_{k \in \bar{\Phi}_{\mathrm{BS}_n} } P_{\mathrm{BS}} |h_k|^2 G_k L(|r_k|) }_{ \mathcal{I}^{\mathrm{BU}}_{n}\{\Phi_{\mathrm{BS}}\} } + \underbrace{ \sum_{t \in \Phi_{\mathrm{RN}} } P_{\mathrm{RN}} |h_t|^2 G_t L(|r_t|) }_{\mathcal{I}^{\mathrm{BU}}_{n}\{\Phi_{\mathrm{RN}}\}} + \sigma^2 },
\end{equation}
where $\Phi_{\mathrm{BS}_n}$ is the point process with density $p_n \lambda_{\mathrm{BS}}$ that cache the file $c_n$. $P_{\mathrm{BS}}$ and $P_{\mathrm{RN}}$ are the transmit power of BS and RN, respectively. $L(|r_{s,d}|) = \gamma (r_{s, d})^{-\alpha}$ is the path loss between node $s$ and node $d$ with $s \in \{\Phi_{\mathrm{BS}},\Phi_{\mathrm{RN}}\}$ and $d \in \{\Phi_{\mathrm{RN}},\Phi_{\mathrm{UE}}\}$. The path loss exponent $\alpha = \alpha_{\mathrm{L}}$ when it is a LOS link and $\alpha = \alpha_{\mathrm{N}}$ when it is an NLOS link. For notational simplicity, let $|r_{s}|$ denote $|r_{s,d}|$ in the denominator.  $\sigma^2$ is the additive Gaussian noise seen at UE$_0$. $\mathcal{I}^{\mathrm{BU}}_{n}\{\Phi_{\mathrm{BS}}\}$ and $\mathcal{I}^{\mathrm{BU}}_{n}\{\Phi_{\mathrm{RN}}\}$ are the inter-cell interferences caused by BSs and RNs, respectively. ${\Phi}_{\mathrm{BS}_n} \backslash \{ \mathrm{BS_0}\}$ is the point process with density $p_n \lambda_{\mathrm{BS}}$ corresponding to the interfering BSs that cache the $n$-th file, and $\bar{\Phi}_{\mathrm{BS}_n} = {\Phi}_{\mathrm{BS}} - {\Phi}_{\mathrm{BS}_n}$ is the point process with density $(1-p_n) \lambda_{\mathrm{BS}}$ corresponding to the interfering BSs that do not cache the $n$-th file.

\textit{2) Two-hop links: }
When UE$_0$ is associated with RN$_0$ that serves as an intermediate node, the SINRs  of the BR link and the RU link can be given respectively as follows,
\begin{equation} \small \label{SINR-def2}
\text{SINR}^{\mathrm{BR}}_{n} = \frac{P_{\mathrm{BS}} |h_n|^2 G_{\mathrm{BS}} L(|r_{\mathrm{BS}_{0}, \mathrm{RN}_{0}}|)}{ \underbrace{ \sum_{j \in \Phi_{\mathrm{BS}_n} \backslash \{BS_0\} } P_{\mathrm{BS}} |h_j|^2 G_j L(|r_j|) + \sum_{k \in \bar{\Phi}_{\mathrm{BS}_n} } P_{\mathrm{BS}} |h_k|^2 G_k L(|r_k|) }_{ \mathcal{I}^{\mathrm{BR}}_{n}\{\Phi_{\mathrm{BS}}\} } + \underbrace{ \sum_{t \in \Phi_{\mathrm{RN}} \backslash \{ \mathrm{RN}_0 \} } P_{\mathrm{RN}} |h_t|^2 G_t L(|r_t|) }_{\mathcal{I}^{\mathrm{BR}}_{n}\{\Phi_{\mathrm{RN}}\}} + \sigma^2 },
\end{equation}
\vspace{-4mm}
\begin{equation} \small \label{SINR-def3}
\text{SINR}^{\mathrm{RU}}_{n} = \frac{P_{\mathrm{RN}} |h_n|^2 G_{\mathrm{RN}} L(|r_{\mathrm{RN}_{0}, \mathrm{UE}_{0}}|)}{ \underbrace{ \sum_{j \in \Phi_{\mathrm{BS}} } P_{\mathrm{BS}} |h_j|^2 G_j L(|r_j|)  }_{ \mathcal{I}^{\mathrm{RU}}_{n}\{\Phi_{\mathrm{BS}}\} } + \underbrace{ \sum_{k \in \Phi_{\mathrm{RN}} \backslash \{ \mathrm{RN}_0 \} } P_{\mathrm{RN}} |h_k|^2 G_k L(|r_k|) }_{\mathcal{I}^{\mathrm{RU}}_{n}\{\Phi_{\mathrm{RN}}\}} + \sigma^2 },
\end{equation}
$\!\!$The rate of each link is given by the Shannon's formula which is expressed as,
\vspace{-3mm}
\begin{equation}\label{Shannon-formula}
  \mathcal{R}_n^{\ell} = B \log_2(1 + \text{SINR}_n^{\ell}),
\end{equation}
where $\ell \in \{ \text{BU, BR, RU} \}$ denotes the link between transmitters and receivers, and $B$ denotes the subchannel bandwidth.

\end{spacing}

\section{Problem Formulation and Analysis }
In order to investigate performance improvement with caching and relaying in the mmWave network, we define SBOP as the performance metric.

\textit{Definition 1 (SBOP)}: The SBOP is defined as the probability that a file requested by UE$_0$ is cached at its associated BS and can be successfully delivered to UE$_0$ at the rate satisfying its requirement, in which case the backhaul traffic is offloaded.

The total SBOP in the relay-assisted mmWave network can be obtained by using the law of total probability, which is given as,
\vspace{-3mm}
\begin{equation}\label{SCDP-def}
  \mathcal{P}_{\mathrm
  {s}}  = \sum_{n=1}^{F} a_n \Big\{ p_{\mathrm{(1hop)}} \mathcal{P}_{\mathrm{s}}^{\mathrm{(1hop)}}(\tau_n ) + p_{\mathrm{(2hop)}} \mathcal{P}_{\mathrm{s}}^{\mathrm{(2hop)}}(\tau_n) \Big\} ,
\end{equation}
where $a_n$ is the probability of requesting file $c_n$, and $\tau_n$ is the target rate of file $c_n$. $p_{\mathrm{(1hop)}}$ and $p_{\mathrm{(2hop)}}$ denote the probability that UE$_0$ is served via the one-hop link (i.e., BS$_0$-UE$_0$) and two-hop link (i.e., BS$_{\mathrm{R}_0}$-RN$_0$-UE$_0$), respectively. $\mathcal{P}_{\mathrm{s}}^{\mathrm{(1hop)}}(\tau_n )$ and $\mathcal{P}_{\mathrm{s}}^{\mathrm{(2hop)}}(\tau_n )$ denote the conditional SBOP of file $c_n$ via the one-hop link and two-hop link, respectively.

\subsection{Analysis of User Association and Relaying}
The triplet BS$_0$, RN$_0$ and BS$_{\mathrm{R0}}$ is identified by using the following user association and relaying criterion:

\vspace{-8mm}
\begin{small}
\begin{subequations}
\begin{align}
   \mathrm{BS}_0 & = \mathop{\arg\max}_{\mathrm{BS}_i \in \Phi_{\mathrm{BS}_n}} \Big\{P_{\mathrm{BS}}  G_{\mathrm{BS}} L(|r_{\mathrm{BS}_i, \mathrm{UE}_0}|)  \Big\},  \label{BS0-rule}  \\
   \mathrm{RN}_0 & = \mathop{\arg\max}_{\mathrm{RN}_j \in \Phi_{\mathrm{RN}}} \left\{ \min \Big\{ P_{\mathrm{RN}}  G_{\mathrm{RN}} L(|r_{\mathrm{RN}_j, \mathrm{UE}_0}|) ,\max_{\mathrm{BS}_k \in \Phi_{\mathrm{BS}_n}} \Big\{ P_{\mathrm{BS}} G_{\mathrm{BS}} L(|r_{\mathrm{BS}_k, \mathrm{RN}_j}|) \Big\}  \Big\}  \right\}, \label{RN0-rule} \\
   \mathrm{BS_{R0}} & = \mathop{\arg\max}_{\mathrm{BS}_i \in \Phi_{\mathrm{BS}_n}} \Big\{P_{\mathrm{BS}}  G_{\mathrm{BS}} L(|r_{\mathrm{BS}_i, \mathrm{RN}_0}|)   \Big\} , \label{BSR0-rule}
\end{align}
\end{subequations}
\end{small}
$\!\!\!$  Eq. (\ref{BS0-rule}) ensures that UE$_0$ receives the highest power from the available BSs that cache the requested file $c_n$. Eq. (\ref{RN0-rule}) ensures that RN$_0$ maximizes the minimum value of the received power of the two-hop link, which is considered to be the best relay. Eq. (\ref{BSR0-rule}) ensures that RN$_0$ receives the highest power from the available BSs that cache the requested file $c_n$.


Whether UE$_0$ is served via a one- or a two-hop link depends on the maximum biased received power, which can be expressed as the following association criterion,
\begin{equation}\label{one-two-hop-condition}
  \begin{cases}
\text{two-hop:} \ \text{BS}_{R0} \to \text{RN}_0 \to \text{UE}_0 \ , \ \text{if} \ \frac{P_{\mathrm{RN}} B_{\mathrm{RN}} G_{\mathrm{RN}}} {|r_{\mathrm{RN}_0, \mathrm{UE}_0}|^{\alpha}}  >  \frac{P_{\mathrm{BS}} B_{\mathrm{BS}} G_{\mathrm{BS}}} {|r_{\mathrm{BS}_0, \mathrm{UE}_0}|^{\alpha}} \ \text{and} \ \frac{P_{\mathrm{BS}} B_{\mathrm{BS}} G_{\mathrm{BS}}} {|r_{\mathrm{BS}_{\mathrm{R0}}, \mathrm{RN}_0}|^{\alpha}}  >  \frac{P_{\mathrm{BS}} B_{\mathrm{BS}} G_{\mathrm{BS}}} {|r_{\mathrm{BS}_0, \mathrm{UE}_0}|^{\alpha}} ,  \\
\text{one-hop:} \ \text{BS}_0 \to \text{UE}_0 \ , \qquad \qquad  \  \text{else},
\end{cases}
\end{equation}
where $B_{\mathrm{BS}}$ and $B_{\mathrm{RN}}$ are bias coefficients that aim to balance the load between one- or two-hop transmission, which is often used in relay-assisted networks \cite{Stochastic-Geometry-Modeling-and-System}.  Specifically, UE$_0$ is served by RN$_0$ only when the maximum biased received power of both the two-hop links (i.e., BR link and RU link) are greater than that of the one-hop link (i.e., BS link).
\textcolor{blue}{Based on the above user association and relaying criterion, it can be inferred that the user association and relaying probabilities is highly dependent on the caching probability $p_n$ and the distance $r$ between different communicating nodes. By mapping the compound random variable $r^{\alpha}$ to the one-dimensional variable $r$ and using (\ref{one-two-hop-condition}), we have
\begin{small}
\begin{subequations}
\begin{flalign}
  & \textcolor{blue}{ \chi_{\mathrm{BU}}(p_n , r) =    \mathbb{P}\left\{  \frac{ |r_{\mathrm{RU}}|^{\alpha} } { \overline{P}_{\mathrm{RN}} } > r  \bigcup \frac{ |r_{\mathrm{BR}}|^{\alpha} } { \overline{P}_{\mathrm{BS}} } > r   \Bigg|  \frac{ |r_{\mathrm{BU}}|^{\alpha} } { \overline{P}_{\mathrm{BS}} } = r \right\},   \label{user_association_chi_BU} }\\
  & \textcolor{blue}{ \chi_{\mathrm{BR}}(p_n, r) =     \mathbb{P}\left\{  \frac{ |r_{\mathrm{BU}}|^{\alpha} } { \overline{P}_{\mathrm{BS}} } > r  \bigcap \frac{ |r_{\mathrm{RU}}|^{\alpha} } { \overline{P}_{\mathrm{BS}} } <  \frac{ |r_{\mathrm{BU}}|^{\alpha} } { \overline{P}_{\mathrm{BS}} }     \Bigg| \frac{ |r_{\mathrm{BR}}|^{\alpha} } { \overline{P}_{\mathrm{RN}} } = r  \right\},  \label{user_association_chi_BR}  }\\
  & \textcolor{blue}{ \chi_{\mathrm{RU}}(p_n , r) =     \mathbb{P}\left\{  \frac{ |r_{\mathrm{BU}}|^{\alpha} } { \overline{P}_{\mathrm{BS}} } > r  \bigcap \frac{ |r_{\mathrm{BR}}|^{\alpha} } { \overline{P}_{\mathrm{BS}} } <  \frac{ |r_{\mathrm{BU}}|^{\alpha} } { \overline{P}_{\mathrm{BS}} }   \Bigg| \frac{ |r_{\mathrm{RU}}|^{\alpha} } { \overline{P}_{\mathrm{RN}} } = r  \right\}.  \label{user_association_chi_RU}  }
\end{flalign}
\end{subequations}
\end{small}
Then the probability density functions (PDFs) of distance between communicating nodes considering the spatial correlation of BSs and RNs should be weighted by the user association and relaying probabilities, which can be obtained by }
\textcolor{blue}{
\begin{small}
\begin{subequations}\label{weighted-distribution}
\begin{align}
  \textcolor{blue}{ \tilde{f}_{\mathrm{BU}}(p_n, r) =  \frac{f_{\mathrm{BS_n}}(r) \chi_{\mathrm{BU}}(p_n, r)}  {\int_0^{\infty} f_{\mathrm{BS_n}}(r) \chi_{\mathrm{BU}}(p_n,r) \mathrm{d} r},   \label{distribution-considering-relay-BU} } \\
  \textcolor{blue}{ \tilde{f}_{\mathrm{BR}}(p_n, r) =  \frac{f_{\mathrm{BS_n}}(r) \chi_{\mathrm{BR}}(p_n, r)}  {\int_0^{\infty} f_{\mathrm{BS_n}}(r) \chi_{\mathrm{BR}}(p_n,r) \mathrm{d} r},  \label{distribution-considering-relay-BR} } \\
  \textcolor{blue}{ \tilde{f}_{\mathrm{RU}}(p_n, r) =  \frac{f_{\mathrm{RN}}(r) \chi_{\mathrm{RU}}(p_n, r)}  {\int_0^{\infty} f_{\mathrm{RN}}(r) \chi_{\mathrm{RU}}(p_n,r) \mathrm{d} r},    \label{distribution-considering-relay-RU} }
\end{align}
\end{subequations}
\end{small}
$\!\!$where $f_{\omega}(r) = 2 \pi \lambda_{\omega} r \mathrm{e}^{- \pi \lambda_{\omega} r^2}, \omega \in \{ \mathrm{BS_n, RN} \}$ are the PDFs of distance between communicating nodes without considering the spatial correlation of BSs and RNs. The denominators in Eqs. (\ref{distribution-considering-relay-BU})-(\ref{distribution-considering-relay-RU}) are the normalized factors.}

\textcolor{blue}{Then the conditional SBOP of file $c_n$ via the one-hop link from BS$_0$ to UE$_0$ at a distance $r$ can be given as
\vspace{-4mm}
\begin{equation}\label{SBOP-expansion-1hop}
  \textcolor{blue}{ \mathcal{P}_{\mathrm{s}}^{(\mathrm{1hop})}(p_n, \tau_n ) = \mathbb{P} \left[\text{SINR}_n^{\mathrm{BU}} > \nu_n \right]  = \int_{r>0} \mathbb{P} \left[\text{SINR}_n^{\mathrm{BU}} > \nu_n | r \right] \tilde{f}_{\mathrm{BU}}(r) \mathrm{d} r, }
\end{equation}
where $\nu_n = 2^{\tau_n / B} - 1$. The conditional SBOP of file $c_n$ via the two-hop links can be given as
\vspace{-3mm}
\begin{align}\label{SBOP-expansion-2hop}
  \textcolor{blue}{ \mathcal{P}_{\mathrm{s}}^{(\mathrm{2hop})}(p_n, \tau_n ) = }& \ \textcolor{blue}{\mathbb{P} \left[\text{SINR}_n^{\mathrm{BR}} > \nu_n, \text{SINR}_n^{\mathrm{RU}} > \nu_n \right]  \nonumber  }\\
  \textcolor{blue}{\overset{(a)} = } & \ \textcolor{blue}{ \int_{r>0} \mathbb{P} \left[\text{SINR}_n^{\mathrm{BR}} > \nu_n | r \right] \tilde{f}_{\mathrm{BR}}(r) \mathrm{d} r \cdot \int_{r>0} \mathbb{P} \left[\text{SINR}_n^{\mathrm{RU}} > \nu_n | r \right] \tilde{f}_{\mathrm{RU}}(r) \mathrm{d} r, }
\end{align}
where (a) is based on the assumption that the three point processes of BSs, RNs and UEs are independent, so $\text{SINR}_n^{\mathrm{BR}}$ and $\text{SINR}_n^{\mathrm{RU}}$ are independent\footnote{\textcolor{blue}{According to the multiplication rule of probability, Eq. (\ref{SBOP-expansion-2hop}) can be written as $  \mathbb{P} \left[\text{SINR}_n^{\mathrm{BR}} > \nu_n, \text{SINR}_n^{\mathrm{RU}} > \nu_n \right]  = \mathbb{P} \left[ \text{SINR}_n^{\mathrm{RU}} > \nu_n \right] \cdot \mathbb{P} \left[\text{SINR}_n^{\mathrm{BR}} > \nu_n \big| \text{SINR}_n^{\mathrm{RU}} > \nu_n \right]$.
Since $\Phi_{\mathrm{BS}}$ is independent of $\Phi_{\mathrm{RN}}$, the channel condition of the BR link is independent of that of the RU link. Thus,
the conditional probability of event $ \text{SINR}_n^{\mathrm{BR}} > \nu_n $ given event $ \text{SINR}_n^{\mathrm{RU}} > \nu_n $ is just the probability of event $ \text{SINR}_n^{\mathrm{BR}} > \nu_n $.}} \cite[Eq. (42)]{Stochastic-Geometry-Modeling-and-System}.
}
\vspace{-2mm}
\textcolor{blue}{\subsection{Problem Formulation}}
\textcolor{blue}{The user association and relaying probabilities fundamentally affects the caching placement probabilities, and vice versa. To jointly optimize the user association and relaying probabilities for different links as well as the caching probabilities of different files at BSs, we formulate the optimization problem as follows,}
\vspace{-3mm}
\textcolor{blue}{
\begin{subequations}
\begin{align}
\textcolor{blue}{\textbf{P1: } \max_{\bm{\chi}, \textbf{p}} } \label{objective-function1}  &\quad \textcolor{blue}{ \mathcal{P}_{\mathrm{s}}  }\\
\textcolor{blue}{ \text{s.t.}  }  &\quad  \textcolor{blue}{ \sum\nolimits_{n=1}^{F} p_n \leq C, \label{constraint:cachesize-constraint1} } \\
&\quad \textcolor{blue}{ 0 \leq p_n \leq 1 ,\quad \forall n\in \mathcal{F},  } \label{constraint:probability-constraint1}  \\
&\quad \textcolor{blue}{ (\ref{user_association_chi_BU}), (\ref{user_association_chi_BR}), (\ref{user_association_chi_RU}), } \nonumber
\end{align}
\end{subequations}}
$\!\!\!$\textcolor{blue}{where $\bm{\chi} = \{\chi_{\mathrm{BU}}, \chi_{\mathrm{BR}}, \chi_{\mathrm{RU}} \}$.  The objective function SBOP in (\ref{objective-function1}) is given by substituting (\ref{SBOP-expansion-1hop}) and (\ref{SBOP-expansion-2hop}) into (\ref{SCDP-def}). }
The constraint (\ref{constraint:cachesize-constraint1}) ensures that the size of the total cached files will not exceed the cache capacity $C$.
\textcolor{blue}{Constraints (\ref{user_association_chi_BU})--(\ref{user_association_chi_RU}) ensure the maximum biased received power criterion.  Problem \textbf{P1} is difficult to solve due to the coupling relationship between the optimization variables $\bm{\chi}$ and $\textbf{p}$ in (\ref{objective-function1}) and (\ref{user_association_chi_BU})--(\ref{user_association_chi_RU}).  Fortunately, with the aid of tools from stochastic geometry,  the optimization variables can be decoupled and problem \textbf{P1} can be transformed into a caching placement optimization problem. }

\textcolor{blue}{ \subsection{Problem Decoupling} }
\textcolor{blue}{To decouple the caching probabilities and the user association and relaying probabilities,  the impact of blockage effects in mmWave networks as well as the spatial correlation caused by the coexistence of BSs and RNs need to be taken into consideration. }
Note that unlike conventional sub-6 GHz cellular networks, the path loss exponents in mmWave networks are different for the LOS or NLOS link. \textcolor{blue}{In the following, we take the blockage effects into account and derive the distribution of inverse biased received power, then the relationship between caching probabilities and the user association and relaying probabilities is provided for problem decoupling. Besides, the impact of blockage effects on the relay-assisted mmWave network is analyzed.}

\begin{lemma}\label{lemma1}
The cumulative distribution function (CDF) of inverse  biased received power from BSs or RNs in the mmWave  network is calculated as
\begin{small}
\begin{align}\label{least-pathloss-distribution}
F_{\omega} (r) =  & \  1 - \mathrm{exp}  \bigg( - \lambda_\omega \Big(  \pi  (r \mathrm{\overline{P}_{\omega}})^{\frac{2}{\alpha_{\mathrm{N}}}}  - \frac{2 \pi }{\beta^2} \left( \vartheta_{\mathrm{L}} -\vartheta_{\mathrm{N}}   \right)  \Big) \bigg), \omega \in \{ \mathrm{BS,RN }\}  ,
\end{align}
\end{small}
where  $\mathrm{\overline{P}_{\omega}} = \textcolor{blue}{\gamma} P_\mathrm{\omega} G_\mathrm{\omega} B_\mathrm{\omega}, G_\mathrm{\omega} = MM$,  $ \vartheta_k = \mathrm{e}^{-\beta(r \mathrm{\overline{P}_{\omega}})^{\frac{1}{\alpha_{k}}}}\left(1+\beta(r \mathrm{\overline{P}_{\omega}})^{\frac{1}{\alpha_{k}}}\right), k \in \{ \mathrm{L, N} \} $, and $\textcolor{blue}{ \mathrm{BS} \in \Phi_{\mathrm{BS_n}}, \mathrm{RN} \in \Phi_{\mathrm{RN}} }$.

\end{lemma}

\begin{proof}
\textcolor{blue}{To prove Lemma 1, we should first obtain the intensity measure  of the non-homogeneous Poisson point process $\Theta_{1} = \left\{\frac{\|r\|^{\alpha}}{\overline{P}_{\omega} }\right\}$ with regard to $r$.
Here, the points of this point process represent the inverse received power from randomly placed BSs or RNs considering the bias coefficients. }
The path loss $\alpha$ is a random variable related to the distance $r$ between communicating nodes, which takes on values $\alpha_{\mathrm{L}}$ and $\alpha_{\mathrm{N}}$ with probability $\mathrm{e}^{-\beta r}$ and $1-\mathrm{e}^{-\beta r}$, respectively.
\textcolor{blue}{By obtaining the intensity measure of Poisson point process $\Theta_{1}$, the distribution of the inverse biased received power $\frac{\|r\|^{\alpha}}{\overline{P}_{\omega} }$ can be obtained by using the void probability $1 - \exp (-\Lambda(0, r))$ of a PPP.}

\textcolor{blue}{
To calculate the intensity measure of $\Theta_{1}$, we utilize the Mapping Theorem \cite[Thm. 2.34]{stochastic-geometry-for-wireless-networks} of the point process transformation. First, the intensity of one-dimensional PPP $\Theta_{2}$ can be calculated as $2 \pi \lambda_{\omega} r$. According to the Mapping Theorem, let $\Phi$ be a PPP on $\mathbb{R}^{s}$ with intensity $\Lambda$ and intensity function $\lambda$, and let $f: \mathbb{R}^{s} \mapsto \mathbb{R}^{d}$ be a measurable function. Then $\Phi^{'} = f(\Phi) = \bigcup_{x \in \Phi } \{ f(x) \}$ is a PPP with intensity measure  $\Lambda^{'}(B) = \Lambda(f^{-1}(B) ) = \int_{f^{-1}(B) } \lambda(x) \mathrm{d} x $ for all $B \subset \mathbb{R}^{d}$. If let $\Phi = \{ ||x||\}$ and  $\Phi^{'} = \Theta_{1} $ , we have $f^{-1} (B) = [0, r]$. Then $B = [0, (r \overline{P}_{\omega})^{\frac{1} {\alpha}}]$. Thus the intensity measure of $\Theta_{1}$ under the effect of blockages can be calculated as}
\begin{equation}\label{mapping-theorem} \small
  \Lambda([0,r]) = \int_0^{(r \overline{P}_{\omega})^{\frac{1}{\alpha_{\mathrm{L}}}}} 2 \pi \lambda_\omega v \mathrm{e}^{-\beta v} \mathrm{d} v + \int_0^{(r \overline{P}_{\omega})^{\frac{1}{\alpha_{\mathrm{N}}}}} 2 \pi \lambda_\omega v (1 - \mathrm{e}^{-\beta v}) \mathrm{d} v.
\end{equation}
Then, the inverse biased received power  distribution is obtained by computing the integral in (\ref{mapping-theorem}) and using the void probability of a PPP.
\end{proof}
\textcolor{blue}{According to Lemma \ref{lemma1} and the maximum biased received power criterion (\ref{user_association_chi_BU})--(\ref{user_association_chi_RU}), the relationship between user association and relaying probabilities and caching probabilities is provided in Proposition \ref{prop-user-association}. }
\textcolor{blue}{
\begin{proposition}\label{prop-user-association}
The user association and relaying probabilities $\bm{\chi} = \{\chi_{\mathrm{BU}}, \chi_{\mathrm{BR}}, \chi_{\mathrm{RU}} \}$ can be expressed as functions of caching probabilities as follows,
\textcolor{blue}{
\begin{subequations}\label{actual-path-loss}
\begin{flalign}
  & \textcolor{blue}{ \chi_{\mathrm{BU}}(p_n , r) =    \overline{F}_{\mathrm{RN}}(r) F_{\mathrm{BS_n}}(r) + \overline{F}_{\mathrm{BS_n}}(r), }  \label{path-loss-distribution-actualBU}  \\
  & \textcolor{blue}{ \chi_{\mathrm{BR}}(p_n, r)  =  \int_{r}^{\infty} \tilde{F}_{\mathrm{BS_n}}(x) F_{\mathrm{RN}}(x) \mathrm{d} x,  } \label{path-loss-distribution-actualBR}   \\
  & \textcolor{blue}{ \chi_{\mathrm{RU}}(p_n , r) =    \int_{r}^{\infty} \tilde{F}_{\mathrm{BS_n}}(x) F_{\mathrm{BS_n}}(x) \mathrm{d} x,   } \label{path-loss-distribution-actualRU}
\end{flalign}
\end{subequations}}
where  $\overline{F}_{\omega}(r)  \triangleq 1 - F_{\omega}(r)$, and $\tilde{F}_{\omega}(r) = \frac{\mathrm{d} F_{\omega}(r)} {\mathrm{d} r}$.
\end{proposition}
\begin{proof}
According to Lemma \ref{lemma1} and  (\ref{user_association_chi_BU}), $\chi_{\mathrm{BU}}(p_n, r)$ can be calculated by using the additive law of probability, which is given by
\begin{align}
\chi_{\mathrm{BU}}(p_n, r) = \overline{F}_{\mathrm{RN}}(r) + \overline{F}_{\mathrm{BS_n}}(r) - \overline{F}_{\mathrm{RN}}(r) \overline{F}_{\mathrm{BS_n}}(r)  = \overline{F}_{\mathrm{RN}}(r) F_{\mathrm{BS_n}}(r) + \overline{F}_{\mathrm{BS_n}}(r).  \nonumber
\end{align}
To calculate $\chi_{\mathrm{BR}}(p_n, r)$, let the random variable $\frac{ |r_{\mathrm{RU}}|^{\alpha} } { \overline{P}_{\mathrm{BS}} } = X$ and $\frac{ |r_{\mathrm{BU}}|^{\alpha} } { \overline{P}_{\mathrm{BS}} } = Y$, then we have
\begin{small}
\begin{align}
& \ \mathbb{P} \left\{ \frac{ |r_{\mathrm{BU}}|^{\alpha} } { \overline{P}_{\mathrm{BS}} } > r  \bigcap \frac{ |r_{\mathrm{RU}}|^{\alpha} } { \overline{P}_{\mathrm{BS}} } <  \frac{ |r_{\mathrm{BU}}|^{\alpha} } { \overline{P}_{\mathrm{BS}} } \right\}  =    \int_{r}^{\infty} \tilde{F}_{\mathrm{BS_n}}(y) \int_{0}^{y} \tilde{F}_{\mathrm{RN}}(x) \mathrm{d} x \mathrm{d} y = \int_{r}^{\infty} \tilde{F}_{\mathrm{BS_n}}(y) F_{\mathrm{RN}}(y) \mathrm{d} y.  \nonumber
\end{align}
\end{small}
Also, $\chi_{\mathrm{RU}}(p_n, r)$ can be calculated by a similar method of calculating $\chi_{\mathrm{BR}}(p_n, r)$. Then the proof is concluded.
\end{proof}}
\textcolor{blue}{
(\ref{path-loss-distribution-actualBU})--(\ref{path-loss-distribution-actualRU}) are the relationships between the user association and relaying probabilities and the caching placement probabilities, which reflect the spatial correlation due to the user association and relaying criterion in (\ref{one-two-hop-condition}).
The distances between communicating nodes are not only affected by the densities of BSs and RNs, but also by whether the UE$_0$ is served via a one- or a two-hop link. For example, when the distance between UE$_0$ and its nearest BS storing the requested file increases, the probability of obtaining the requested file via one-hop transmission decreases, which will affect the actual distance distribution between a UE and its associated BS or RN.
To this end, it is meaningful to derive the partial derivative of $\left| \chi_{\ell}(p_n, \beta) \right|, \ell \in \{ \text{BU, BR, RU} \}$ with respect to $r$, and investigate the impact of blockage effects or the caching probabilities on the changing trend of  user association and relaying probabilites given $r = r_0$, which is denoted by
$\widetilde{\chi}_{\ell}(p_n, \beta) =  \frac{\partial \, \chi_{\ell} (p_n, r)} {\partial \,  r} \big| _{r = r_0}  $.
}
\textcolor{blue}{
\begin{corollary}\label{limit-prop}
Assume that no link is established through  NLOS links, then $\left| \widetilde{\chi}_{\mathrm{BU}}(p_n, \beta) \right|$ decreases monotonically with the blockage parameter $\beta$ and increases monotonically with the caching probability $p_n$, and   $\lim_{\beta \rightarrow \infty } \tilde{f}_{\mathrm{BU}}(p_n, r) =  f_{\mathrm{BS_n}}(r)$ holds.
\end{corollary}
\begin{proof}
We first derive the first-order derivative of $F_{\mathrm{BS_n}}(r)$  without considering the NLOS transmission by omitting the second term in (\ref{mapping-theorem}) , which is given as
\begin{small}
\begin{align}
   \tilde{F}_{\mathrm{BS_n}} (r)  = \frac{\mathrm{d} F_{\mathrm{BS_n}}(r)} {\mathrm{d}  r} = &    \frac{2 \pi \lambda_\omega }{\alpha_{\mathrm{L}}} \overline{P}_{\omega}^{\frac{2}{\alpha_{\mathrm{L}}}} r^{\frac{2}{\alpha_{\mathrm{L}}} - 1} \exp \left(\frac{2 \pi \lambda_\omega  \vartheta_{\mathrm{L}}}{\beta^2} - \beta (r \overline{P}_{\omega} )^{\frac{1}{\alpha_{\mathrm{L}}}}\right),
\end{align}
\end{small}
where $\vartheta_{\mathrm{L}} = \mathrm{e}^{-\beta(r \mathrm{\overline{P}_{\omega}})^{\frac{1}{\alpha_{\mathrm{L}}}}}\left(1+\beta(r \mathrm{\overline{P}_{\omega}})^{\frac{1}{\alpha_{\mathrm{L}}}}\right)$, then the partial derivative of $\vartheta $ with respect to $\beta$ can be calculated as
$\frac{\partial \vartheta_{\mathrm{L}}}{\partial \beta} = - \beta \mathrm{e}^{- \beta (r \overline{P}_{\omega} )^{\frac{1}{\alpha_{\mathrm{L}}}}  } (r \overline{P}_{\omega} )^{\frac{2}{\alpha_{\mathrm{L}}}} < 0$. Thus $\tilde{F}_{\mathrm{BS_n}} (r)$ is monotonically decreasing with $\beta$, and monotonically increasing with $p_n$. Applying a similar method and we can obtain the same conclusion for $\tilde{F}_{\mathrm{RN}} (r)$. Then according to (\ref{path-loss-distribution-actualBU}), it can be derived that  $\tilde{\chi}_{\ell} (p_n,  \beta) = \tilde{\overline{F}}_{\mathrm{RN}}(r)  F_{\mathrm{BS_n}}(r) + \overline{F}_{\mathrm{RN}}(r) \tilde{F}_{\mathrm{BS_n}}(r)  + \tilde{\overline{F}}_{\mathrm{BS_n}}(r) | _{r = r_0} $ , where $ \tilde{\overline{F}}_{\mathrm{RN}}(r), \tilde{\overline{F}}_{\mathrm{BS_n}}(r) <0$, $0 \leq F_{\mathrm{BS_n}}(r), \overline{F}_{\mathrm{RN}}(r) \leq 1$, and $\tilde{\overline{F}}_{\mathrm{BS_n}}(r) = - \tilde{F}_{\mathrm{BS_n}}(r)$. Thus $\tilde{\chi}_{\mathrm{BU}}(p_n, \beta) < 0$ and  $| \tilde{\chi}_{\mathrm{BU}}(p_n, \beta) |$ is monotonically decreasing with $\beta$, and monotonically increasing with $p_n$.
When $\beta \rightarrow \infty $, it can be inferred that $|\tilde{\chi}_{\ell} (p_n,  \beta)| \rightarrow 0$, then $|\chi_{\ell} (p_n,  \beta)| $ should be a constant. According to the weighted PDFs in (\ref{distribution-considering-relay-BU}), $\lim_{\beta \rightarrow \infty } \tilde{f}_{\mathrm{BU}}(p_n, r) =  f_{\mathrm{BS_n}}(r)$ holds.
\end{proof}
Corollary \ref{limit-prop} shows that the absolute value of gradient of $\chi_{\mathrm{BU}}$ in the direction of $r$ at $r_0$ decreases monotonically with $\beta$. This implies that for $\beta \rightarrow 0$, the average distance between one-hop communicating nodes is minimized, which maximizes the advantage of introducing relays in mmWave networks, i.e., more communicating nodes can transmission data through two-hop links. While for $\beta \rightarrow \infty$, the PDFs of distance between communicating nodes will remain unchanged, i.e., $\tilde{f}_{\mathrm{BU}}(p_n, r) =  f_{\mathrm{BS_n}}(r)$, which implies that deploying relay nodes cannot improve the system performance in this situation. This reflects a trade-off between mmWave frequency and the effect of relays in mmWave networks. Adopting higher mmWave frequency can obtain larger bandwidth, but at the expense of deterioration of mitigating blockage effects by relays.
}

\textcolor{blue}{Substituting (\ref{path-loss-distribution-actualBU})--(\ref{path-loss-distribution-actualRU}) into (\ref{objective-function1}), problem \textbf{P1} can be transformed into a caching placement optimization problem as follows, }
\textcolor{blue}{
\begin{subequations}
\begin{align}
\textcolor{blue}{\textbf{P2: } \max_{\textbf{p}} } \label{objective-function}  &\quad \textcolor{blue}{\mathcal{P}_{\mathrm{s}}  }\\
\textcolor{blue}{ \text{s.t.}   } &\quad \textcolor{blue}{ \sum\nolimits_{n=1}^{F} p_n \leq C, } \label{constraint:cachesize-constraint}  \\
&\quad \textcolor{blue}{0 \leq p_n \leq 1 ,\quad \forall n\in \mathcal{F}. }  \label{constraint:probability-constraint}
\end{align}
\end{subequations}}
\vspace{-10mm}
\textcolor{blue}{
\subsection{Derivation of SBOP   }}
In the following, we aim to derive the expression of the conditional SBOP $\mathcal{P}_{\mathrm{s}}^{\mathrm{(1hop)}}$ and $\mathcal{P}_{\mathrm{s}}^{\mathrm{(2hop)}}$.


\begin{proposition}\label{prop-PBS}
The SBOP in the relay-assisted mmWave network when UE$_0$ requests the $n$-th file and is served via a one-hop link is given by

\vspace{-4mm}
\textcolor{blue}{
\begin{small}
\begin{align}\label{prop-equation-PBS}
  \mathcal{P}_{\mathrm{s}}^{\mathrm{(1hop)}}(p_n, \tau_n ) \approx & \ \sum_{i_1 = 1}^{q_1} w_{i_1} \mathrm{e}^{r_{i_1}}  W^{\mathrm{BS}}(p_n, \tau_n, r_{i_1})  \mathcal{L}_{\mathcal{I}_n^{\mathrm{BU}} }(p_n, r_{i_1})        \tilde{f}_{\mathrm{BU}}(p_n, r_{i_1}) ,
\end{align}
\end{small}
$\!\!\!$where  $w_{i_1} = \frac{r_{i_1}} {(q_1 + 1)^2 [L_{q_1 + 1}(r_{i_1})]^2 }$, $r_{i_1}$ is the $i_1$-th zero of $L_{q_1}(r)$, $L_{q_1}(r)$ denotes the Laguerre polynomials, and $q_1$ is a  parameter balancing the accuracy and complexity.  $W^{\omega}(p_n, \tau_n, r) = \sum_{k \in \{ \mathrm{L,N}\}} \rho_k(r) \sum_{u=1}^{N_k} (-1)^{u+1} \binom{N_k}{u}  \mathrm{e}^{ -\frac{u \eta_k \nu_n r^{\alpha_k}  \sigma^2 }{P_{\omega} G_{0}}}   $,  $G_0 = MM$,  $\eta_k = N_k (N_k !)^{-\frac{1}{N_k}}$,     and
\begin{small}
\begin{align}\label{Laplace}
  &  \mathcal{L}_{\mathcal{I}_n^{\mathrm{BU}} }(p_n, r)   =  \prod_{i \in \{ \mathrm{L,N}\}} \prod_{G} Q_i^{\mathrm{BS}}(p_n, r)  Q_i^{\mathrm{BS}}(\overline{p}_n, 0)  Q_i^{\mathrm{RN}}(1, 0),
\end{align}
\end{small}
\begin{small}
\begin{align}\label{Q-r-def}
   Q_i^{\omega}(p_n, r)  & = \mathrm{exp} \left[ -2 \pi p_n \lambda_{\omega} p_G \sum_{u=1}^{N_i} \binom{N_i}{u} \frac{r^{-\frac{1}{\alpha_i} \left(u - \frac{2}{\alpha_i}\right) }} {u \alpha_{i} - 2}  {_2}\textit{F}_{1}  \left(N_i, u - \frac{2}{\alpha_i}; 1 + u - \frac{2}{\alpha_i}; -s   r^{-\frac{1}{\alpha_i}} \right)    \right],
\end{align}
\end{small}
\begin{small}
\begin{align}\label{Q-0-def}
   Q_i^{\omega}(p_n, 0)  & = \mathrm{exp} \left[ -\frac{2}{\alpha_i} \pi p_n \lambda_{\omega} p_G \sum_{u=1}^{N_i} \binom{N_i}{u} s^{-(u - \frac{2}{\alpha_i})}  \textit{B}\left(u - \frac{2}{\alpha_i}, N_i - u + \frac{2}{\alpha_i}\right)     \right],
\end{align}
\end{small}
$\!\!$where $s = \frac{u \eta_k \nu_n  P_{\omega} G} {r^{-\alpha_k} N_i}  $, $G \in \{ MM, Mm, mm \} $ with probability $p_G$, $\overline{p}_n = 1 - p_n$, ${_2}\textit{F}_{1} (\cdot, \cdot; \cdot; \cdot) $ is the Gauss hypergeometric function, and $B(\cdot, \cdot)$ is the Beta function.}
\end{proposition}
\begin{proof}
Please refer to Appendix \ref{AppendixB}.
\end{proof}

Likewise, the data transmission via a two-hop link is considered to be successful if both the data rates of BS$_0$-RN$_0$ link and RN$_0$-UE$_0$ link are greater than the target data rate of the $n$-th file $\tau_n$. Using a similar way as the proof of Proposition \ref{prop-PBS}, we have the following proposition.
\textcolor{blue}{
\begin{proposition}\label{prop-PRN}
The SBOP in the relay-assisted mmWave network when UE$_0$ is served via a two-hop link aided by RN$_0$ is given by
\begin{small}
\begin{align}\label{prop-equation-PRN}
& \ \mathcal{P}_{\mathrm{s}}^{\mathrm{(2hop)}}(\{\tau_n \})   \nonumber \\
  \approx & \ \sum_{i_1 = 1}^{q_1}  \sum_{j_1 = 1}^{q_1} w_{i_1} w_{j_1} \mathrm{e}^{r_{i_1} + r_{j_1}}  W^{\mathrm{BS}}(p_n, \tau_n, r_{i_1})    \mathcal{L}_{\mathcal{I}_n^{\mathrm{BR}} }(p_n, r_{i_1})      \tilde{f}_{\mathrm{BR}}(p_n, r_{i_1})         W^{\mathrm{RN}}(p_n, \tau_n, r_{j_1})     \mathcal{L}_{\mathcal{I}_n^{\mathrm{RU}} }(p_n, r_{j_1})     \tilde{f}_{\mathrm{RU}}(p_n, r_{j_1}),
\end{align}
\end{small}
\end{proposition}
where
\begin{small}
\begin{align}\label{Laplace-2hop}
&  \mathcal{L}_{\mathcal{I}_n^{\mathrm{BR}} }(p_n, r) = \prod_{i \in \{ \mathrm{L,N}\}} \prod_{G} Q_i^{\mathrm{BS}}(p_n, r)  Q_i^{\mathrm{BS}}(\overline{p}_n, 0), \\
&  \mathcal{L}_{\mathcal{I}_n^{\mathrm{RU}} }(p_n, r)  =  \prod_{i \in \{ \mathrm{L,N}\}} \prod_{G} Q_i^{\mathrm{RN}}(1, r) Q_i^{\mathrm{BS}}(1, 0).
\end{align}
\end{small}
}
\textcolor{blue}{
The approximation in (\ref{prop-equation-PBS}) and (\ref{prop-equation-PRN}) adopts the Alzer's approximation \cite{On-some-inequalities}, which takes the complementary CDF of a gamma random variable as a weighted sum of the CDFs of exponential random variables.  Note that $Q_i^{\omega}(p_n, r)$  in  (\ref{prop-equation-PBS}) and (\ref{prop-equation-PRN}) represents the Laplace transform of the interference from all the other BSs or RNs. We can observe that the interference signal dynamically changes with the caching placement probability $p_n$. Thus, both the user association considering the spatial correlation of nodes and the dynamic network interference should be optimally controlled to obtain the optimal caching placement solution.}


On the other hand, if we assume that mmWave transmissions are noise-limited \cite{Coverage-in-heterogeneous, Tractable-model-for-rate}, the expression of SBOP will become more tractable, which can also provide some design insights into the caching placement strategy.

\begin{proposition}\label{prop4}
The SBOP in the relay-assisted mmWave network can be transformed into a closed-form expression in the noise-limited scenario, which is given by
\begin{align}\label{SBOP-simple}
  \mathcal{P}_{\mathrm{s}}^{\mathrm{NL}} = \sum_{n=1}^{F}  K_n a_n \Big(1 - \exp \big( - p_n T_n \big) \Big),
\end{align}
where $K_n = p_{\mathrm{1hop}} +  \left( 1 - \exp \left(  - c_{\mathrm{RN}} \left(\xi_{\mathrm{RN}}\right)^{\kappa_{\mathrm{N}}} - Y_n\left(\xi_{\mathrm{RN}}\right) \right) \right) p_{\mathrm{2hop}}$,  $c_{\omega} =  \pi \lambda_{\omega} \kappa_{\mathrm{N}}  \frac{\Gamma(\kappa_{\mathrm{N}} + N_\mathrm{L})}{(N_\mathrm{L})^{\kappa_{\mathrm{N}}} \Gamma(N_\mathrm{L})}$, $\kappa_{\mathrm{N}} = \frac{2}{\alpha_{\mathrm{N}}}$, $\Gamma(\cdot)$ is the gamma function, $\xi_{\omega} = \frac{P_{\omega} G_{\omega} }{\sigma^2 \nu_n}$, $T_n = c_{\mathrm{BS}} \left(\xi_{\mathrm{BS}}\right)^{\kappa_{\mathrm{N}}} + Y_n (\xi_{\mathrm{BS}})$, and
\begin{equation}\label{Yj-def}
  Y_n(\xi_\omega) = \sum_{k \in \{ \mathrm{L, N} \}}  \frac{2\pi (N_k)^{N_k} \lambda_{\omega}}{\alpha_k \Gamma(N_k)}   \int_0^{\infty} \int_0^{\xi_\omega} \psi^{(2/\alpha_k+N_k-1)} \mathrm{e}^{\big(-\frac{N_k}{\hat{\xi}}\psi - \beta \psi^{\frac{1}{\alpha_k}} \big)} \big/ \hat{\xi}^{(N_k+1)}  \mathrm{d} \hat{\xi} \mathrm{d} \psi.
\end{equation}
\end{proposition}
\begin{proof}
Please refer to Appendix \ref{AppendixC}.
\end{proof}
Note that (\ref{SBOP-simple}) is a convex function with respect to caching probability vector $\textbf{p}$.
\textcolor{blue}{In addition, the value of (\ref{SBOP-simple}) increases slower as $p_n$ increases, which can reflect a trade-off between caching the most popular files and file diversity of cached files. Specifically, caching popular files with larger probability $p_n$ will shorten the average distance $r$ between communicating nodes, thus beneficial for SBOP performance. However, this benefit will increase less as $\lambda_{\mathrm{RN}}$ increases, because it is more likely to be a LOS link for a shorter $r$, resulting in a better channel condition.
In this case, considering that the caching capacity is limited, reducing the caching probability of some files and increasing those of other files may result in a higher SBOP performance.}


Next, we calculate the probabilities that UE$_0$ is served via a one-hop link or two-hop link. Based on the CDF of reverse maximum biased received power derived in Lemma 1 and the property of the quotient of random variables, the probability that UE$_0$ is served via a two-hop link can be calculated as,
\begin{small}
\begin{align}\label{two-hop-condition}
  p_{\mathrm{(2hop)}} = & \mathbb{P} \left\{ \frac {|r_{\mathrm{RU}}|^{\alpha}} {\overline{P}_{\mathrm{RN}}} <  \frac {|r_{\mathrm{BU}}|^{\alpha}} {\overline{P}_{\mathrm{BS}} } \bigcap \frac {|r_{\mathrm{BR}}|^{\alpha}}  {\overline{P}_{\mathrm{BS}} } <  \frac {|r_{\mathrm{BU}}|^{\alpha}}  {\overline{P}_{\mathrm{BS}} } \right\}  =  \frac{1}{2}  \int_{0}^{\infty} \tilde{F}_{\mathrm{BS}}(r) F_{\mathrm{RN}}(r) \mathrm{d} r.
\end{align}
\end{small}
$\!\!$   Otherwise, a typical UE is served via a one-hop link with probability $p_{\mathrm{(1hop)}} = 1 - p_{\mathrm{(2hop)}}$.

In the next section, two caching placement algorithms will be proposed to maximize SBOP.

\section{Caching Placement Optimization}

The main difficulty to solve problem \textbf{P2} is that the expression of SBOP does not have a closed form and it is a non-convex optimization problem due to the existence of binomial term $(-1)^{m+1}$. Fortunately, we find that an algorithm based on monotonic optimization can be proposed to optimally solve this problem.

\subsection{Optimal caching placement algorithm based on monotonic optimization}
\textit{1) Monotonic optimization}:
First, we introduce some mathematical definitions that will be useful for monotonic optimization \cite{Monotonic, Monotonic-optimization-in}.

\textit{Definition 2 (Box)}: If $\mathbf{a} \preceq \mathbf{b}$, then box $[\mathbf{a}, \mathbf{b}]$ is the set of all $\mathbf{z} \in \mathbb{R}^{n}$ satisfying $\mathbf{a} \preceq \mathbf{z} \preceq \mathbf{b}$.

\textit{Definition 3 (Normal set)}: A set $\mathcal{Z} \subset \mathbb{R}_{+}^n $ ($\mathbb{R}_{+}$ denotes the set of non-negative real numbers) is normal if for any element $\mathbf{z} \in \mathcal{Z}$, all other elements $\mathbf{z}'$ such that $\mathbf{0}\preceq \mathbf{z}' \preceq\mathbf{z}$ are also in set $\mathcal{Z}$.

\textit{Definition 4 (Projection)}: Given any non-empty normal set $\mathcal{Z} \subset \mathbb{R}_{+}^n$ and any vector $\mathbf{z} \in \mathbb{R}_{+}^n\setminus \mathcal{Z}$, $\mathbf{\Upsilon}(\textbf{z})$ is the projection of $\mathbf{z}$ onto the boundary of $\mathcal{Z}$, i.e., $\mathbf{\Upsilon} (\mathbf{z}) = \gamma_0  \mathbf{z}$, where $\gamma_0 = \max \{ \gamma > 0 | \gamma \mathbf{z} \in \mathcal{Z} \}$.

\textit{Definition 5}: An optimization problem belongs to the class of monotonic optimization problems if it can be formulated in the following form,
\vspace{-4mm}
\begin{align}\label{monotonic-def}
& \max_{\mathbf{z}}    \quad \psi (\mathbf{z})  \\
 & \quad \text{s.t.} \quad  \mathbf{z} \in \mathcal{Z}, \nonumber
\end{align}
where $\psi (\mathbf{z})$ is an increasing function on $\mathbb{R}_{+}^n$ and set $\mathcal{Z} \subset \mathbb{R}_{+}^n$ is a non-empty normal set.

\begin{spacing}{1.35}
\textit{2) Optimal caching placement algorithm}:
To apply the monotonic optimization, we first show that the objective function in (\ref{objective-function}) is an increasing function with respect to $\mathbf{p}$. Notice that in (\ref{prop-equation-PBS}), the items containing the optimization variables $\mathbf{p}$ are $Q_i^{\mathrm{BS}}(p_n, r)$,  $ Q_i^{\mathrm{BS}}(\overline{p}_n, 0) $ and $\tilde{f}_{\mathrm{BU}}(r)$. Since the CDF of $f_{\mathrm{BU}}(r)$ is $1 - \mathrm{e}^{- p_n \lambda_{\mathrm{BS}} \pi r^2}$, and the weighting factor $\chi_{\mathrm{BU}}$ further reduces the average distance between communicating nodes, thus the CDF of   $\tilde{f}_{\mathrm{BU}}(r)$ is monotonically increasing with $\mathbf{p}$.  In addition, taking the derivative of $Q_i^{\mathrm{BS}}(p_n, r) Q_i^{\mathrm{BS}}(\overline{p}_n, 0) $, we have
\begin{equation}\label{derivative-QQ}
  \frac{\partial Q_i^{\mathrm{BS}}(p_n, r) Q_i^{\mathrm{BS}}(\overline{p}_n, 0)}{\partial p_n} =  U_i(r) \mathrm{exp} \left[  \left(  p_n U_i(r) - U_i(\infty) \right) \right]  > 0,
\end{equation}
where $U_i(r) = 2 \pi \lambda_{\mathrm{BS}} p_G \int_{0}^{r} \left( 1 - \frac{1}{(1 + s_k P_{\mathrm{BS}} G t^{-\alpha_i}/N_i)^{N_i}}\right) \rho_{i}(t) t \mathrm{d} t$. Thus, we can conclude that the conditional SBOP where UE$_0$ is served via a one-hop link $\mathcal{P}_{\mathrm{s}}^{\mathrm{BS}}(\{\tau_n \})$ is monotonically increasing with $p_n$. Similarly, we can get the conclusion that the conditional SBOP where UE$_0$ is served via a two-hop link $\mathcal{P}_{\mathrm{s}}^{\mathrm{RN}}(\{\tau_n \})$ is also monotonically increasing with $p_n$.

With the analysis above, the problem \textbf{P2} can be written as a standard monotonic optimization problem as shown in (\ref{monotonic-def}), which is expressed as follows,
\vspace{-3mm}
\begin{align}\label{convert2monotonic}
\textbf{P3: } & \max_{\mathbf{z}}    \quad \psi (\textbf{z})=\mathcal{P}_{\mathrm{s}}  \\
 & \quad \text{s.t.} \quad  \mathbf{z} \in \mathcal{Z}, \nonumber
\end{align}
where $\mathcal{Z}=\{\mathbf{z}|0\leq z_i \leq p_i,(\ref{constraint:cachesize-constraint}), (\ref{constraint:probability-constraint}),i=1,\cdots,F\}$. $\psi (\textbf{z})$ is an increasing function on $\mathbb{R}_{+}^F$ and feasible set $\mathcal{Z}$ is a non-empty normal set. We now can design the caching placement algorithm to solve the monotonic optimization problem in (\ref{convert2monotonic}) based on the polyblock outer approximation approach \cite{Monotonic}. A polyblock $\mathcal{B} \subset \mathbb{R}_{+}^F $ is the union of all the boxes $[\mathbf{0},\mathbf{z}]$, $\mathbf{z}\in \mathcal{Z}$, where $\mathcal{Z}$ is the vertex set of the polyblock. A polyblock is clearly a normal set. According to \cite{Monotonic}, since the objective function in (\ref{convert2monotonic}) is a monotonically increasing function, the optimal solution is always obtained at the boundary of the feasible set $\mathcal{Z}$. Therefore, the basic idea of polyblock outer approximation is to construct a sequence of polyblocks to approach the boundary of the feasible set with increasing accuracy. Then, the optimal solution exists at one vertex of a polyblock.

\renewcommand{\algorithmicrequire}{\textbf{Input:}}
\renewcommand{\algorithmicensure}{\textbf{Output:}}

\begin{algorithm}
\caption{\mbox{Caching placement algorithm based on polyblock outer approximation (CP-POA)}}
\label{alg:PA}
\LinesNumbered
\KwIn {$\mathcal{P}_{\mathrm{s}}$, $F$, $C$, $\epsilon$;}
\KwOut {Optimal solution $\textbf{p}^{*}=(p^*_i)_{i\in \mathcal{F}}$;}
\textbf{initialization:} Set the iteration index $k=1$. Construct the polyblock $\mathcal{B}^{(1)}$ with vertex set $\mathcal{V}^{(1)} = \{ \textbf{v}^{(1)}\}$, where the entries of $\textbf{v}^{(1)}$ are set as $v_i^{(1)} = 1(i=1,\cdots,F)$;

\Repeat {$\frac{\left\|\mathbf{v}^{(k)}-\mathbf{\Upsilon}\left(\mathbf{v}^{(k)}\right)\right\|}{\left\|\mathbf{v}^{(k)}\right\|} \leq \epsilon$}
{
 $k = k + 1$;

 Compute $\mathbf{\Upsilon}\left(\mathbf{v}^{(k-1)}\right)$ based on Definition 3;

 Construct a smaller polyblock $\mathcal{B}^{(k)}$ with the vertex set $\mathcal{V}^{(k)}$ by replacing $\mathbf{v}^{(k-1)}$ in $\mathcal{V}^{(k-1)}$ with $F$ new vertices $\big\{\tilde{\mathbf{v}}_{1}^{(k-1)}, \cdots, \tilde{\mathbf{v}}_{F}^{(k-1)}\big\}$, where the $i$-th new vertex is obtained by
\begin{equation*}\label{alg-equ1}
  \tilde{\mathbf{v}}_{i}^{(k-1)}=\mathbf{v}^{(k-1)}+\left(\Upsilon_{i}\left(\mathbf{v}^{(k-1)}\right)-v_{i}^{(k-1)}\right) \mathbf{e}_{i}.
\end{equation*}

 Find $\textbf{v}^{(k)}$ as that vertex from $\mathcal{V}^{(k)}$ whose projection maximizes the objective function \\ of the problem (\ref{convert2monotonic}), i.e., $\mathbf{v}^{(k)}=\arg \max \left\{\psi(\mathbf{\Upsilon}(\textbf{v}) ) | \mathbf{v} \in \mathcal{V}^{(k)}\right\}$;
}

The optimal solution is obtained as $\boldsymbol{p}^{*} = \mathbf{\Upsilon} \left(\mathbf{v}^{(k)}\right)$.
\end{algorithm}

The construction method of polyblocks works as follows. First, we construct a polyblock $\mathcal{B}^{(1)}$ that encloses the feasible set $\mathcal{Z}$ with vertex set $\mathcal{V}^{(1)}$. $\mathcal{V}^{(1)}$ includes only one vertex $\textbf{v}^{(1)}$, where the entries of $\textbf{v}^{(1)}$ can be initialized as $v_i^{(1)}=1,(i=1,\cdots,F)$.
According to Definition 4, the projection of $\textbf{v}^{(1)}$ on the boundary of $\mathcal{Z}$ can be found by bisection search, denoted as $\mathbf{\Upsilon}(\textbf{v}^{(1)})$.
Then a smaller polyblock $\mathcal{B}^{(2)}$ is constructed based on $\mathcal{B}^{(1)}$ by replacing $\textbf{v}^{(1)}$ with $F$ new vertices $\tilde{\mathcal{V}}^{(1)}=\left\{\tilde{\mathbf{v}}_{1}^{(1)}, \cdots, \tilde{\mathbf{v}}_{F}^{(1)}\right\}$. Thus we can get new vertice set $\mathcal{V}^{(2)}=\left(\mathcal{V}^{(1)} \backslash \mathbf{v}^{(1)}\right) \cup \tilde{\mathcal{V}}^{(1)}$, which constitutes $\mathcal{B}^{(2)}$ that still encloses $\mathcal{Z}$.
Note that the new vertex $\tilde{\mathbf{v}}_{i}^{(1)}$ is generated by replacing the $i$-th entry of $\mathbf{v}^{(1)}$ with the $i$-th entry of $\mathbf{\Upsilon} (\mathbf{v}^{(1)})$, which is given by
\begin{equation}\label{replace-formula}
  \tilde{\mathbf{v}}_{i}^{(1)}=\mathbf{v}^{(1)}+\left(\Upsilon_{i}\left(\mathbf{v}^{(1)}\right)-v_{i}^{(1)}\right) \mathbf{e}_{i},
\end{equation}
where $\Upsilon_{i}\left(\mathbf{v}^{(1)}\right)$ is the $i$-th entry of $\mathbf{\Upsilon}(\textbf{v}^{(1)})$, and $\mathbf{e}_{i}$ is the $i$-th unit vector of $\mathbb{R}^{F\times 1}$ with a non-zero entry only at index $i$. Then, we choose the optimal vertex from $\mathcal{V}^{(2)}$ whose projection maximizes the objective function of the problem (\ref{convert2monotonic}), i.e., $\mathbf{v}^{(2)}=\arg \max \left\{\psi(\mathbf{\Upsilon}(\mathbf{v}) ) | \mathbf{v} \in \mathcal{V}^{(2)}\right\}$.
Repeating this procedure, we can construct a sequence of polyblocks that gradually outer approximate the feasible set, i.e.,
\begin{equation}\label{similar-procedure}
  \mathcal{B}^{(1)} \supset \mathcal{B}^{(2)} \supset \cdots \supset \mathcal{B}^{(k)} \supset \cdots \supset \mathcal{Z}.
\end{equation}


The algorithm terminates when $\frac{\left\|\mathbf{v}^{(k)}-\mathbf{\Upsilon}\left(\mathbf{v}^{(k)}\right)\right\|}{\left\|\mathbf{v}^{(k)}\right\|} \leq \epsilon$, where $\epsilon\geq0$ is the given error tolerance specifying the accuracy of the approximation.
The algorithm is outlined in Algorithm \ref{alg:PA} (CP-POA). From the optimal solution $\boldsymbol{p}^{*}$ obtained with CP-POA, we can get the optimal caching placement.

According to \cite{Monotonic}, the convergence of CP-POA is guaranteed if the monotonic optimization problem like (\ref{monotonic-def}) satisfies the following three conditions: $\psi (\textbf{z})$ is upper semicontinuous, $\mathcal{Z}$ has a nonempty interior and $\mathcal{Z} \in \mathbb{R}^n_{++}$ ($\mathbb{R}_{++}$ denotes the set of positive real numbers). Considering problem (\ref{convert2monotonic}), we can see that $\psi (\textbf{z})$ is a continuous function, $\mathcal{Z}$ has a nonempty interior but $\mathcal{Z} \in \mathbb{R}^{F}_{+}$. To guarantee convergence, we can define $\mathbf{y} = \mathbf{z} + \mathbf{1} $ so that $\mathcal{Z}=\{\mathbf{y}|1\leq y_i \leq p_i+1,(\ref{constraint:cachesize-constraint}), (\ref{constraint:probability-constraint}),i=1,\cdots,F\}$ satisfying $\mathcal{Z} \in \mathbb{R}^{F}_{++}$. Correspondingly, the initialization in Algorithm \ref{alg:PA} is changed to $v_i^{(1)} = 2(i=1,\cdots,F)$.
Besides, from \cite{Monotonic, Monotonic-optimization-in}, we know that the optimal solution of monotonic optimization problems such as problem (\ref{convert2monotonic}) can be obtained via the CP-POA algorithm.
\textcolor{blue}{However, the computational complexity increases exponentially with the number of vertices $F$ generated in each iteration. In the following, we propose a suboptimal caching placement algorithm to strike a balance between computational complexity and system performance.}

%


\subsection{Suboptimal caching placement algorithm based on convex optimization}
The existing literature has shown that mmWave transmissions tend to be noise-limited and the interference is weak \cite{Coverage-in-heterogeneous, Tractable-model-for-rate}, this is due to the fact that in the presence of blockages, the signals received from unintentional sources are close to negligible. Hence, in this subsection we investigate caching placement in the relay-assisted mmWave network under the noise-limited scenario. In this case, we can derive the closed-form expression of SBOP presented in (\ref{SBOP-simple}), and then propose a suboptimal caching placement algorithm using convex optimization (CP-CO).

Under the noise-limited scenario, the objective function in the optimization problem \textbf{P2} can be transformed into a convex function, which is written as
\vspace{-3mm}
\begin{align}
\textbf{P4: } \max_{\bm{p}}   &\quad \mathcal{P}_{\mathrm{s}}^{\mathrm{NL}} = \sum_{n=1}^{F}  K_n a_n \Big(1 - \exp \big( - p_n T_n \big) \Big)   \\
&\quad \text{s.t.} \quad (\text{\ref{constraint:cachesize-constraint}}), (\text{\ref{constraint:probability-constraint}}) ,  \nonumber
\end{align}
The Lagrangian function of this optimization problem is
\begin{align}\label{Lagrangian-func}
  & \mathcal{L} (\{p_n\},  \varepsilon, \{\mu_n  \})  =  - \sum_{n=1}^{F}  K_n a_n \Big(1 - \exp \big( - p_n T_n \big) \Big)  +  \varepsilon (\sum_{n=1}^{F} p_n - C) + \sum_{n=1}^{F} \mu_n (p_n - 1),
\end{align}
where $\varepsilon, \mu_n$ are the Lagrangian multipliers associated with the constraints (\text{\ref{constraint:cachesize-constraint}}), (\text{\ref{constraint:probability-constraint}}), respectively.
This constrained optimization problem can be solved by applying the Karush-Kuhn-Tucker (KKT) conditions. After differentiating $\mathcal{L} (\{p_n\},  \varepsilon, \{\mu_n  \})$ with respect to $p_n$, we can obtain all the necessary KKT conditions for the optimal caching probability, which is given by
\vspace{-6mm}

\begin{small}
\begin{subequations}
\begin{numcases}{}
   \frac{\partial \mathcal{L} (\{p_n\},  \varepsilon, \{\mu_n  \}) }{\partial p_n} = -K_n a_n T_n \exp(-p_n T_n) + \varepsilon + \mu_n \geq 0, \label{KKT1} \\
   \big( -K_n a_n T_n \exp(-p_n T_n) + \varepsilon + \mu_n \big) p_n = 0,  \label{optimal-condition} \\
   \varepsilon \left(\sum\nolimits_{n=1}^{F} p_n - C\right) = 0,  \\
   \mu_n(p_n-1) = 0,  \\
   \varepsilon, \mu_n \geq 0. \label{KKT5}
\end{numcases}
\end{subequations}
\end{small}
Then, the optimal caching probability is derived from the constraint in (\ref{optimal-condition}), which is given by
\begin{equation}\label{optimal-caching-prob} \small
  p_n = \left[\frac{1}{T_n} \ln \left(\frac{K_n a_n T_n}{\varepsilon + \mu_n}\right)\right]^{+},
\end{equation}
where $[x]^{+} = \max \{ 0,x \}$. The caching probability of file $c_n$ increases as the file popularity $a_n$ becomes larger, but is regulated by the Lagrangian multipliers $\ln(\varepsilon + \mu_n)$. According to the KKT conditions in (\ref{KKT1}) - (\ref{KKT5}), the Lagrangian multipliers $\varepsilon$ and $\mu_n$ range with respect to $p_n$, which is given by
\begin{equation} \label{piecewise-func} \small
\begin{cases}
  \varepsilon \leq K_n a_n T_n \mathrm{e}^{-T_n}, \quad \qquad \qquad \mu_n = [K_n a_n T_n \mathrm{e}^{- T_n} - \varepsilon ]^{+}  \ \text{for} \ p_n = 1,  \\
  K_n a_n T_n \mathrm{e}^{-T_n} \leq \varepsilon \leq K_n a_n T_n , \ \mu_n = 0, \ \text{for} \ 0 < p_n < 1,  \\
  \varepsilon \geq K_n a_n T_n , \quad \qquad \qquad \qquad \mu_n = 0, \ \text{for} \ p_n = 0.
  \end{cases}
\end{equation}

(\ref{piecewise-func}) indicates that the caching probability $p_n$ is determined by the Lagrangian multiplier $\varepsilon$ since $\mu_n$ is a function of $\varepsilon$. Specifically, if $\varepsilon \leq \min \{ K_n a_n T_n \mathrm{e}^{-T_n}| n \in \mathcal{F} \}$, then $p_n = 1, n \in \mathcal{F}$ and $\sum_{n=1}^{F}p_n(\varepsilon, \mu_n ) = F$; if $\varepsilon \geq \max \{ K_n a_n T_n | n \in \mathcal{F} \}$, then $p_n = 0, n \in \mathcal{F}$ and $\sum_{n=1}^{F}p_n(\varepsilon, \mu_n ) = 0$; when $\min \{ K_n a_n T_n \mathrm{e}^{-T_n}| n \in \mathcal{F} \} \leq \varepsilon \leq \max \{ K_n a_n T_n | n \in \mathcal{F} \} $, $\sum_{n=1}^{F}p_n(\varepsilon, \mu_n ) $ is bounded by $0 \leq \sum_{n=1}^{F}p_n(\varepsilon, \mu_n ) \leq F$ since $\sum_{n=1}^{F}p_n(\varepsilon, \mu_n ) $ is decreasing with $\varepsilon$. Hence, due to the fact that $\sum_{n=1}^{F}p_n(\varepsilon^{*}, \mu_n^{*} ) = C$, the optimal Lagrangian multiplier $\varepsilon^{*}$ can be efficiently found by the bisection method, which is presented in Algorithm \ref{alg:bisection}.

\textcolor{blue}{The complexity of Algorithm \ref{alg:bisection} in steps 5-10 is logarithmic due to the bisection method, i.e., $\mathcal{O} \left(\log (K_n T_n) \right)$. Observing (\ref{SBOP-simple}) and (\ref{Yj-def}), we find the proportional relationship between  $K_n$, $T_n$ and the scale of the problem, i.e., $K_n \propto \exp(\lambda_{\mathrm{BS}} + \lambda_{\mathrm{RN}})$, $T_n \propto (\lambda_{\mathrm{BS}} + \lambda_{\mathrm{RN}})$. Thus the complexity of the bisection method is $\mathcal{O} (\lambda_{\mathrm{BS}} + \lambda_{\mathrm{RN}}) $. Since steps 2-4 is the inner loop of the bisection method, the overall complexity of Algorithm \ref{alg:bisection} is $\mathcal{O} (F (\lambda_{\mathrm{BS}} + \lambda_{\mathrm{RN}}))$.}

\begin{algorithm} \label{alg:bisection}
\caption{A bisection method for finding $\{p_n^{*}\}$}
\LinesNumbered
\KwIn{$ a_n, K_n, T_n , n \in \mathcal{F}$;}
\KwOut{the optimal caching probabilities $\{p_n^{*}\}$;}
\textbf{Initialization:} {$l \leftarrow \min \{ K_n a_n T_n \mathrm{e}^{-T_n}| n \in \mathcal{F} \}, u \leftarrow \max \{ K_n a_n T_n | n \in \mathcal{F} \}, \varepsilon \leftarrow \frac{l+u}{2}$;}

\For {$n = 1,2,\cdots,F$ }
{ \label{bisection-step2}
   $\mu_n \leftarrow [K_n a_n T_n \mathrm{e}^{-T_n} - \varepsilon]^{+}$;

   Calculate $p_n(\varepsilon, \mu_n)$ according to (\ref{optimal-caching-prob}); \label{bisection-step4}
}

\While {$|\sum_{n=1}^{F}p_n(\varepsilon, \mu_n) - C | \geq \epsilon$}
{
    \uIf{$|\sum_{n=1}^{F}p_n(\varepsilon, \mu_n) > C $}
    {$l \leftarrow \varepsilon$ ;}

    \uElseIf{$|\sum_{n=1}^{F}p_n(\varepsilon, \mu_n) < C $}
    { $u \leftarrow \varepsilon$ ;}
   $\varepsilon \leftarrow \frac{l+u}{2}$, then update $\mu_n$ by repeating steps \ref{bisection-step2}-\ref{bisection-step4};
}

$p_n^{*} = p_n(\varepsilon, \mu_n)$, for $\forall n \in \mathcal{F}$;
\end{algorithm}

\subsection{BS and RN selection algorithm}
After determining the caching probabilities by the proposed algorithms in the caching placement phase as mentioned above, we give the BS and RN selection algorithm during the file delivery phase. In this paper, we consider the node selection algorithm depending on both cached files and maximum biased received power. Before the BS and RN selection algorithm terminates, the node that may be selected as the associated node is termed as the candidate node.

\begin{algorithm} \label{selection-algorithm}
\caption{BS and RN selection algorithm}
\LinesNumbered
\KwIn{$ \Phi_{\mathrm{BS}}, \Phi_{\mathrm{RN}}, P_{\omega}, B_{\omega}, G_{\omega}, c_n, p_n,\omega \in \{ \text{BS, RN} \}, n \in \mathcal{F} $;}
\KwOut{The selection of the associated BS$_{R0}$, RN$_0$ or BS$_0$ for UE$_0$;}
\textbf{Initialization:} {set maximum biased received power $P_{\mathrm{BS_0}} = 0$, $P_{\mathrm{RN_0}} = 0$, $P_{\mathrm{BS}_\mathrm{R0}} = 0$;}

\mbox{Each BS independently caches files based on $\{ p_n\}$ obtained by Algorithm 1 or Algorithm 2.}

UE$_0$ searches for the BSs that cache the requested file $c_n$, and the set is denoted as $\Phi_{\mathrm{BS}_n}$;

\If{ $\Phi_{\mathrm{BS}_n} = \emptyset$ }
{
    fail in offloading backhaul for $c_n$ and go to step \ref{fail-offloading};
}

Select BS$_0$ from $\Phi_{\mathrm{BS}_n}$ based on (\ref{BS0-rule}); \label{step5}

Eliminating the RNs from $\Phi_{\mathrm{RN}}$ whose biased received power are below $P_{\mathrm{BS_0}}$, and denote the remaining set as $\Phi_{(\mathrm{RN>BS_0})}$;

\ForEach {$\mathrm{RN}_i \in \Phi_{(\mathrm{RN>BS_0})}$ }
{ \label{step8}
    Select the $\mathrm{BS_{R0}}$ based on (\ref{BSR0-rule});

    \If {$\min \{ P_{\mathrm{RN}_{i}}, P_{\mathrm{BS}_{\mathrm{R0}}} \} > P_{\mathrm{RN}_{0}} $}
    {
         $P_{\mathrm{RN}_{0}} \Leftarrow \min \{ P_{\mathrm{RN}_{i}}, P_{\mathrm{BS}_\mathrm{R0}} \} $ , record the corresponding index of RN$_0$ and BS$_\mathrm{R0}$;  \label{step10}
    }
}

\eIf {$P_{\mathrm{RN}_{0}} < P_{\mathrm{BS}_{0}}$}
{
    \eIf{the rate requirement of $c_n$ is satisfied by the link BS$_{0}$-UE$_0$}
    {
        UE$_0$ is served by BS$_0$ via the one-hop link;
    }
    {
       fail in offloading backhaul for $c_n$ and go to step \ref{fail-offloading};
    }
}
{
    \eIf{the rate requirement of $c_n$ is satisfied by the link BS$_\mathrm{R0}$-RN$_0$-UE$_0$ }
    {
        UE$_0$ is served by RN$_0$ and BS$_\mathrm{R0}$ via the two-hop link;
    }
    {
        fail in offloading backhaul for $c_n$ and go to step \ref{fail-offloading};
    }

     \mbox{Repeat steps \ref{step5}-\ref{step10} by replacing $\Phi_{\mathrm{BS}_n}$ in step \ref{step5} with $\Phi_{\mathrm{BS}}$ to reselect BS$_0$, RN$_0$ and BS$_\mathrm{R0}$.} \label{fail-offloading}
}
\end{algorithm}

The proposed BS and RN selection algorithm is outlined in Algorithm \ref{selection-algorithm}. In the caching placement phase, each BS independently fetches files via the backhaul and caches them according to the caching probabilities obtained by the proposed CP-POA or CP-CO algorithm. In the file delivery phase, the associated BS is preferentially selected from the BSs that cache the requested file by UE$_0$. Whether UE$_0$ is served via a one-hop link or a two-hop link is based on the relaying and association criterion described in Section \ref{association-criterion}. Specifically, the BS caching the requested file $c_n$ and having the largest biased received power is selected as the candidate BS$_0$ first, and the corresponding biased received power is denoted by $P_{\mathrm{BS_0}}$. Next, according to Eq. (\ref{one-two-hop-condition}), RNs with lower biased received power than BS$_0$ is excluded from the candidate RN$_0$. Then, the candidate RN$_0$ and BS$_\mathrm{R0}$ is determined by the loop in steps \ref{step8}-\ref{step10}, which is in line with Eqs. (\ref{RN0-rule}) and (\ref{BSR0-rule}), maximizing the minimum of $P_{\mathrm{RN}_{0}}$ and $P_{\mathrm{BS}_\mathrm{R0}}$. Finally, whether UE$_0$ is served via a one- or a two-hop link is determined based on the obtained maximum biased received power of the candidate BS$_0$, RN$_0$ and BS$_\mathrm{R0}$. In addition, whether the traffic is successfully offloaded from the backhaul depends on whether the rate requirement is satisfied.
\textcolor{blue}{The complexity of Algorithm \ref{selection-algorithm} is dominated by steps \ref{step8}-\ref{step10}, where the outside loop depends on  $|\Phi_{(\mathrm{RN>BS_0})} |$, and the inner loop depends on $|\Phi_{\mathrm{BS}} |$. Thus the worst-case complexity of Algorithm \ref{selection-algorithm} is $\mathcal{O} (\lambda_{\mathrm{BS}} \lambda_{\mathrm{RN}})$.  }

Note that when the file requested by UE$_0$ is not cached in any BS or the rate requirement is not satisfied, UE$_0$ will select the associated BS/RN only based on the maximum biased received power, which is not ideal for backhaul offloading.  The purpose of the caching placement algorithms we designed in the previous subsections is to reduce the probability of occurrence of this operating case, and will be verified in the next section.




\section{Performance Evaluation}

In this section, we validate our analytical work and evaluate the performance of SBOP using the proposed caching placement algorithms, CP-POA and CP-CO. Unless otherwise stated, most of the system parameters and their corresponding values are given in Table \ref{simulationsetting}. We conduct both numerical experiments and Monte Carlo simulations in various scenarios. In the Monte Carlo simulations, the performance is averaged over 2000 network deployments, where all the BSs, RNs and UEs are randomly distributed in a square area of 800 m $\times$ 800 m.
For the sake of comparison, \textcolor{blue}{ we consider three benchmark algorithms as follows}: 1) caching $C$ most popular files (MPC) \cite{Caching-policy-toward-maximal}, 2) caching files uniformly (UC) \cite{Caching-at-the-wireless-edge}, \textcolor{blue}{and 3) constrained cross-entropy optimization (CCEO) \cite{Content-placement-in-cache-enabled}, which is a heuristic algorithm based on stochastic searching.  }
\begin{table}
  \centering
  \caption{Parameter Values}
  \label{simulationsetting}
  \small
  \begin{tabular}{|p{2.2cm}|p{9.8cm}|p{2.9cm}|}
    \hline
    \textbf{Parameters}   & \textbf{Physical meaning}        & \textbf{Values}\\\hline
     $P_{\mathrm{BS}}$  &   Transmit power of each BS      &30 dBm           \\\hline
     $P_{\mathrm{RN}}$  &   Transmit power of each RN      &30 dBm           \\\hline
     $B$  & The bandwidth assigned to each UE &  100 MHz  \\\hline
     $\alpha_{\mathrm{L}} / \alpha_{\mathrm{N}}$  & Path loss exponent of LoS and NLoS  & 2.5 / 4     \\\hline
     $\theta$  & Mainlobe beamwidth  &  30$^\circ$      \\\hline
     $M / m$  & Mainlobe antenna gain / sidelobe antenna gain & 10 dB / -10 dB  \\\hline
     $\beta$   &  Blockage density  &  4$\times$10$^{-4}$   \\\hline
     $N_\mathrm{L} / N_\mathrm{N}$   &   Nakagami fading parameter for LoS and NLoS channel & 3 / 2  \\\hline
     $\lambda_{\mathrm{BS}}$  &  mmWave BS density  &  10$^{-5}$ nodes/m$^2$   \\\hline
     $\lambda_{\mathrm{RN}}$  &  mmWave RN density  &  10$^{-5}$ nodes/m$^2$   \\\hline
     $\tau_n, \forall n \in \mathcal{F}$  &  Target data rate for the $n$-th file  &  0.04$-$1 Gbps   \\\hline
     $\delta$  &  Skewness of the file popularity   &  0.8  \\\hline
     $F$  &  The number of files  &  20  \\\hline
     $C$  &  Maximum number of files cached by each BS  &  10  \\\hline
  \end{tabular}
\end{table}

We begin by validating the convergence of the proposed CP-POA algorithm for different number of files and cache size. As can be seen from Fig. 2, the CP-POA algorithm converges to the optimal solution for different network parameters. We note that CP-POA starts from an infeasible initial probability decision policy which violates some of the constraints, but may yield a high SBOP. Ultimately CP-POA converges to the feasible optimal solution. \textcolor{blue}{In general, the convergence of the proposed CP-POA algorithm becomes slower as the number of files increases. Moreover, extensive simulations suggest that the speed of convergence is not sensitive to other network parameters (e.g., blockage parameter, caching size, etc.).}
\begin{figure}
  \begin{minipage}[h]{0.29\linewidth}
  \includegraphics[width=1.05\textwidth]{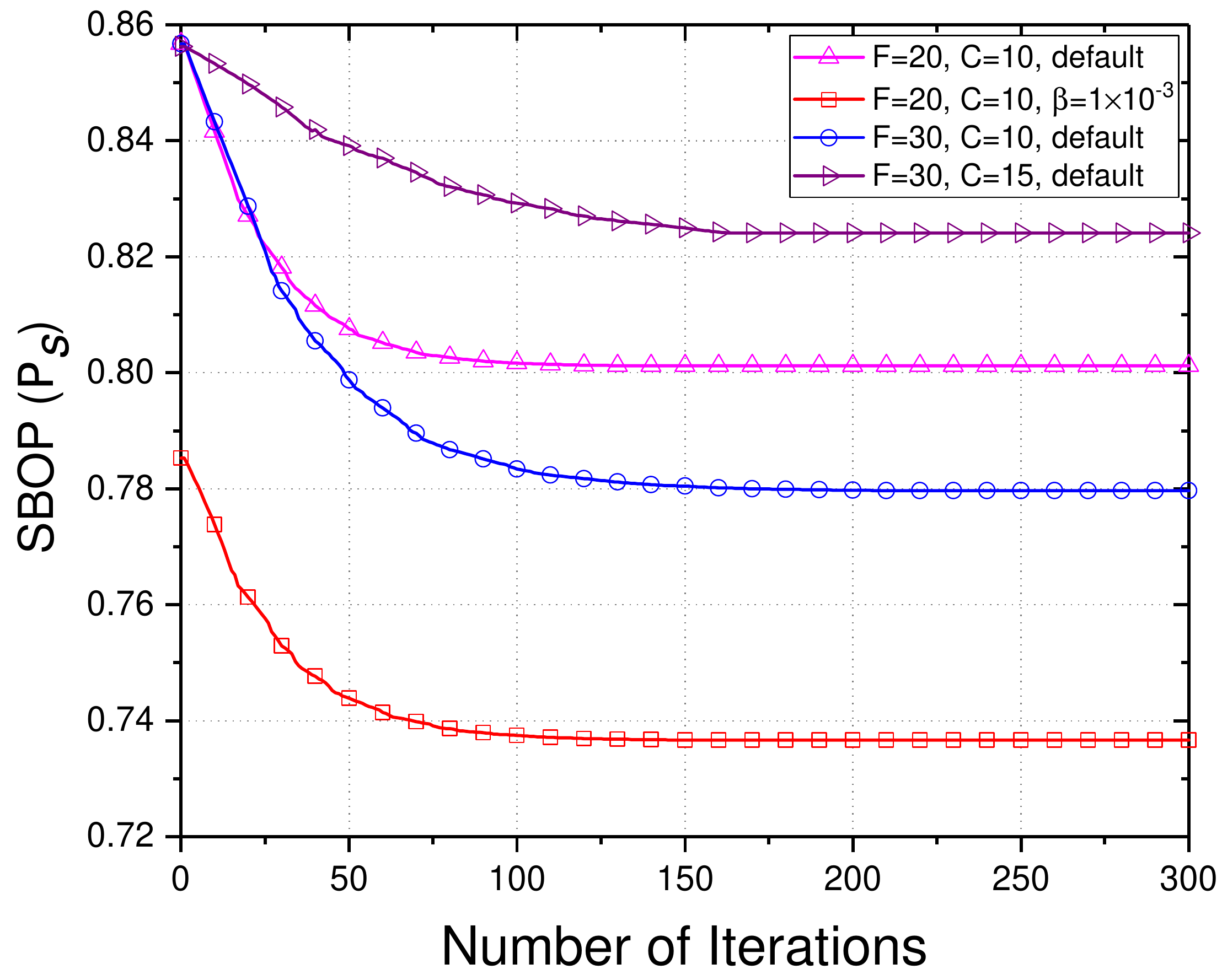}
  \label{fig:SBOP-converge}
  \textcolor{blue}{\caption{Convergence of the proposed CP-POA algorithm for different $F$ and $C$. }}
  \end{minipage}
  \hspace{0.28in}
  \begin{minipage}[h]{0.70\linewidth}
  \subfigure[]{
    \label{fig:mmw-pn-lambdaR} 
    \includegraphics[width=0.44\textwidth]{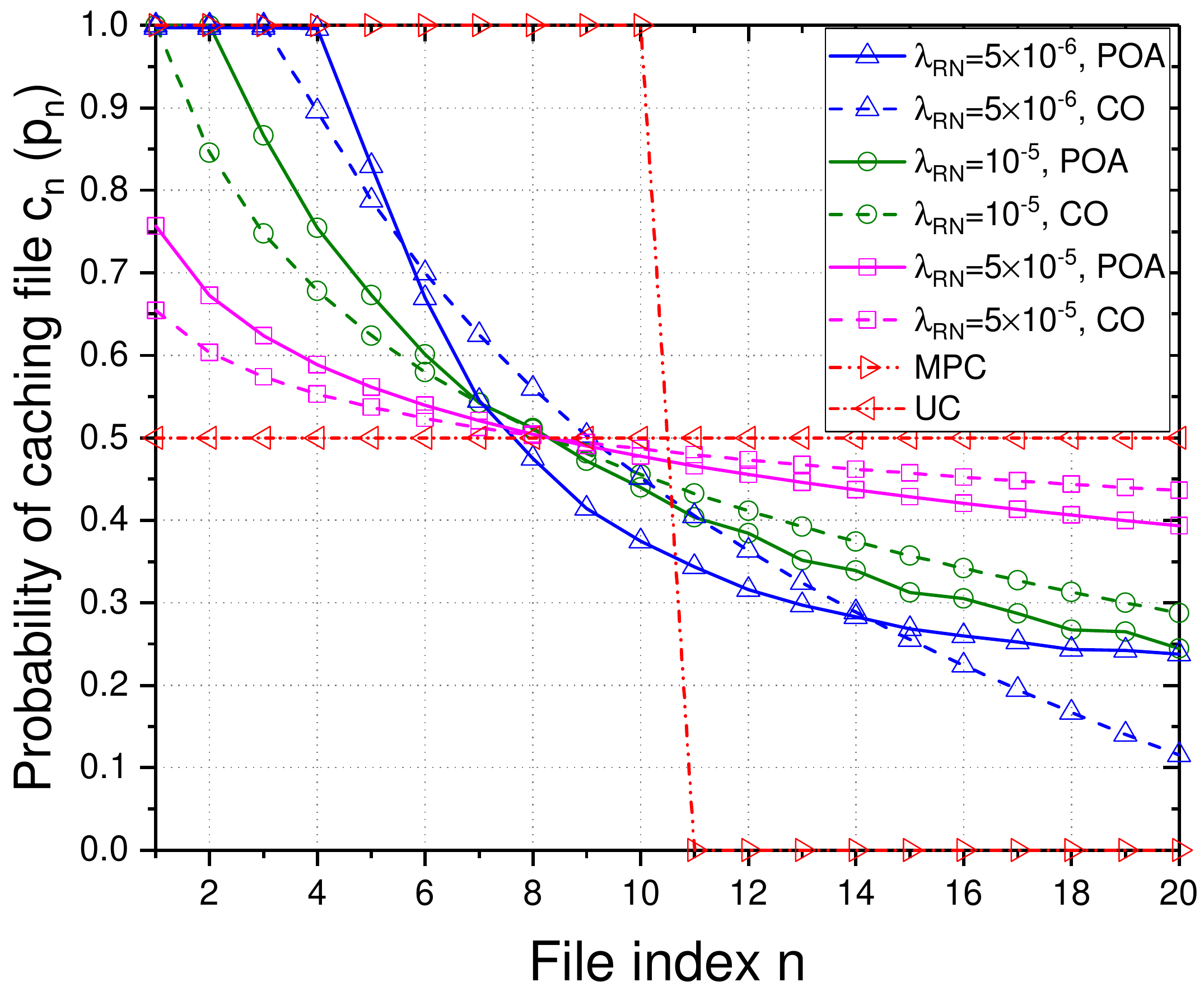}}
  \subfigure[]{
    \label{fig:mmw-pn-delta} 
    \includegraphics[width=0.44\textwidth]{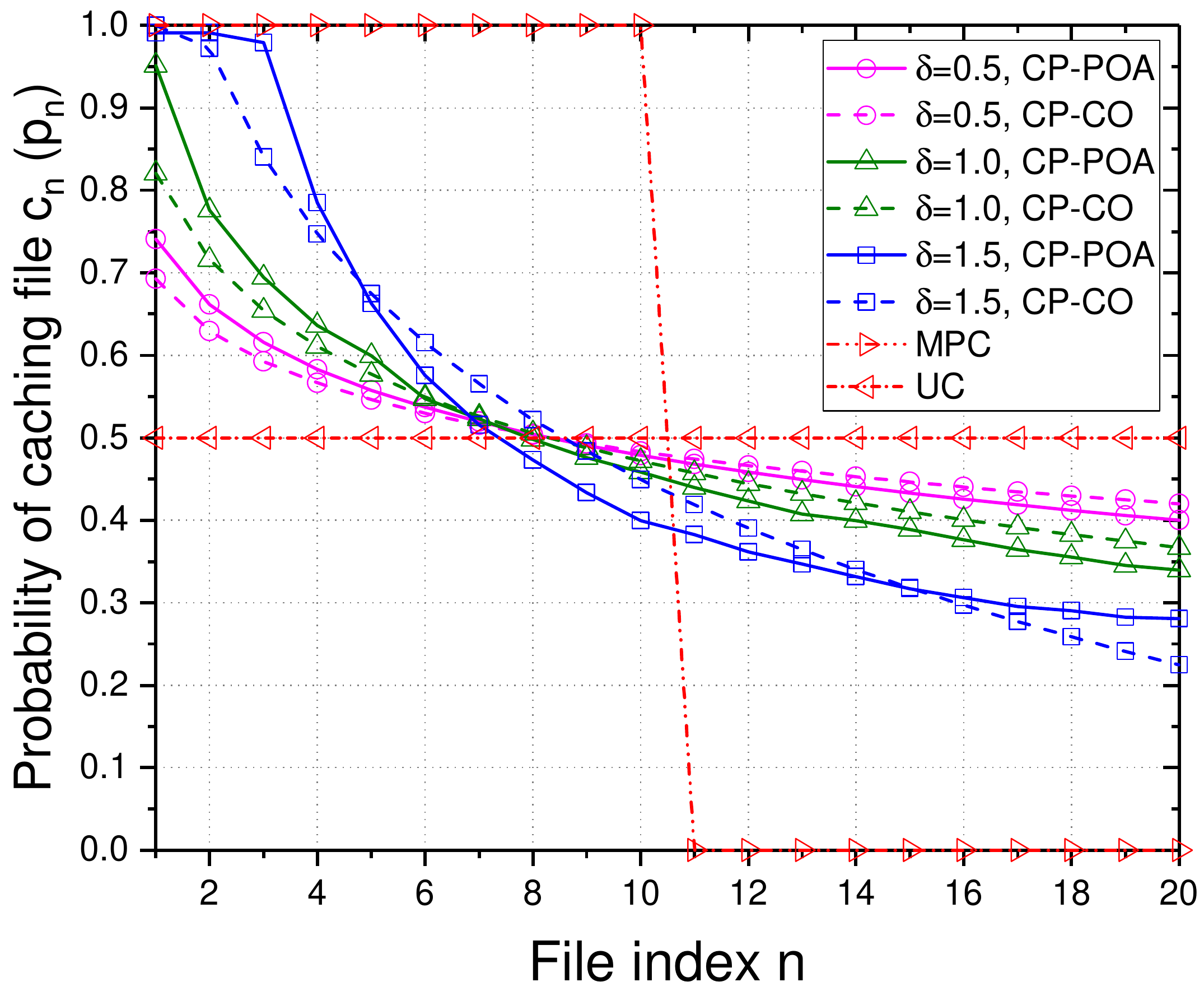}}
  \textcolor{blue}{\caption{Results of the proposed algorithms v.s. file index for (a) different densities \newline of RNs $\lambda_{\mathrm{RN}}$, and (b) different skewnesses of file popularity $\delta$.}}
  \label{fig:mmw-pn} 
  \end{minipage}
\end{figure}


We now evaluate the performance of the proposed caching placement algorithms.
We first evaluate the optimal caching probabilities for varying densities of RNs as shown in Fig. \ref{fig:mmw-pn-lambdaR}. Note that these files are sorted by the popularity, i.e., a file with smaller index owns higher popularity.
In general, when the density of RNs is relatively low, the most popular files are cached with a higher probability using the proposed algorithms. In contrast, when the density of RNs is high, the caching probabilities for the most popular files decrease, and the optimal caching probabilities become more uniform. According to the thinning theorem, the higher the caching probability of the most popular files, the more BSs will cache them, thus offering shorter geometric communication distance for the specific file request, which is dominant when the density of RNs is low. However, with the increase of RN density, the average geometric communication distance will be further shortened, so it is better to increase the caching probabilities for the less popular files so that the file diversity in the network can be increased. Besides, in Fig. \ref{fig:mmw-pn-delta}, we also evaluate the effect of Zipf parameter $\delta$ reflecting the skewness of file popularity. It can be seen that the caching probabilities tend to be more uniformly distributed with a smaller $\delta$. This is because the requests for files are more decentralized with a smaller $\delta$, in which case considering file diversity is more beneficial than caching popular files.
\begin{figure}
  \centering
  \subfigure[]{
    \label{fig:SBOP-lambdaRN} 
    \includegraphics[width=0.3\textwidth]{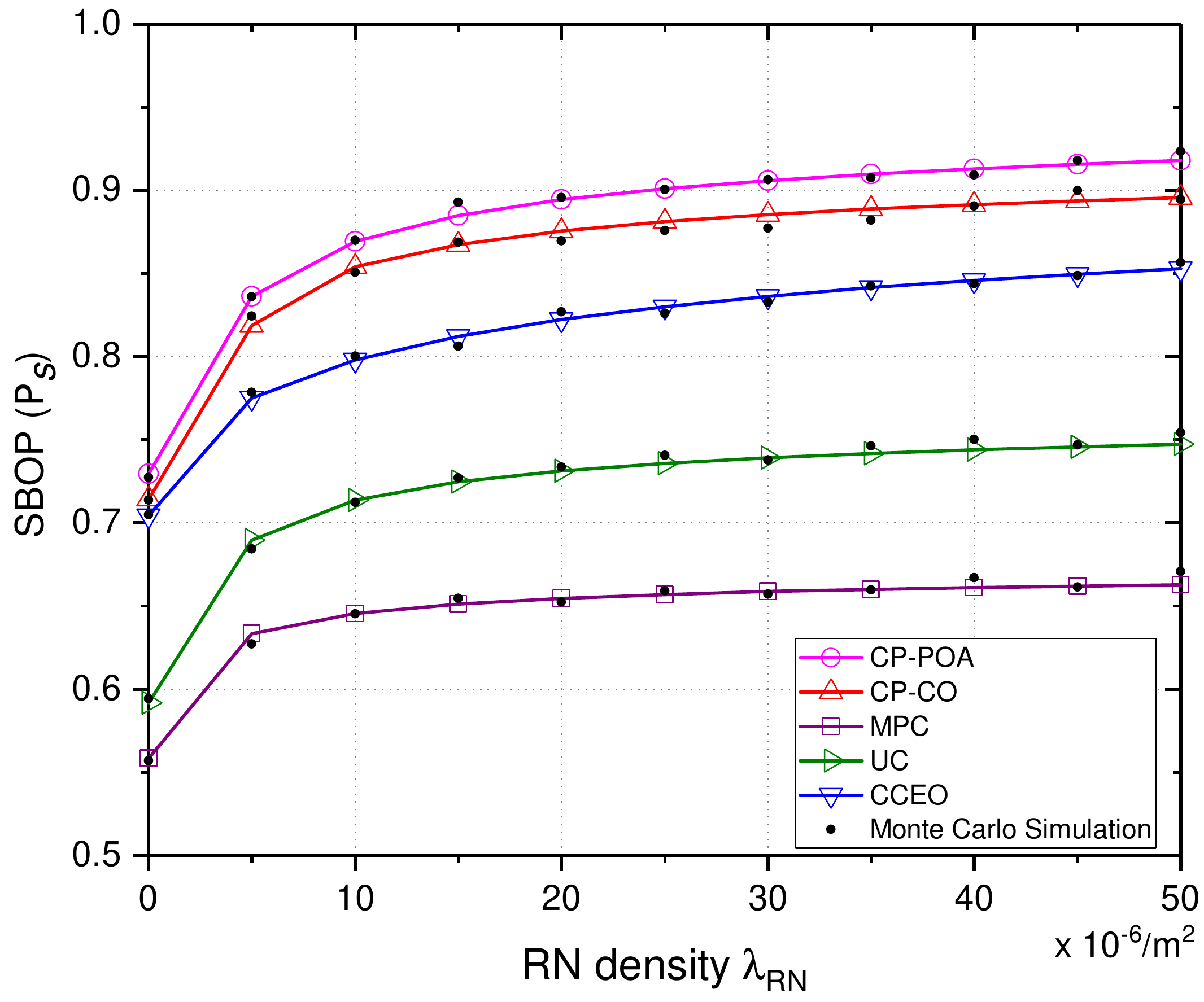}}
  \hspace{0.05in}
  \subfigure[]{
    \label{fig:SBOP-SINRt} 
    \includegraphics[width=0.295\textwidth]{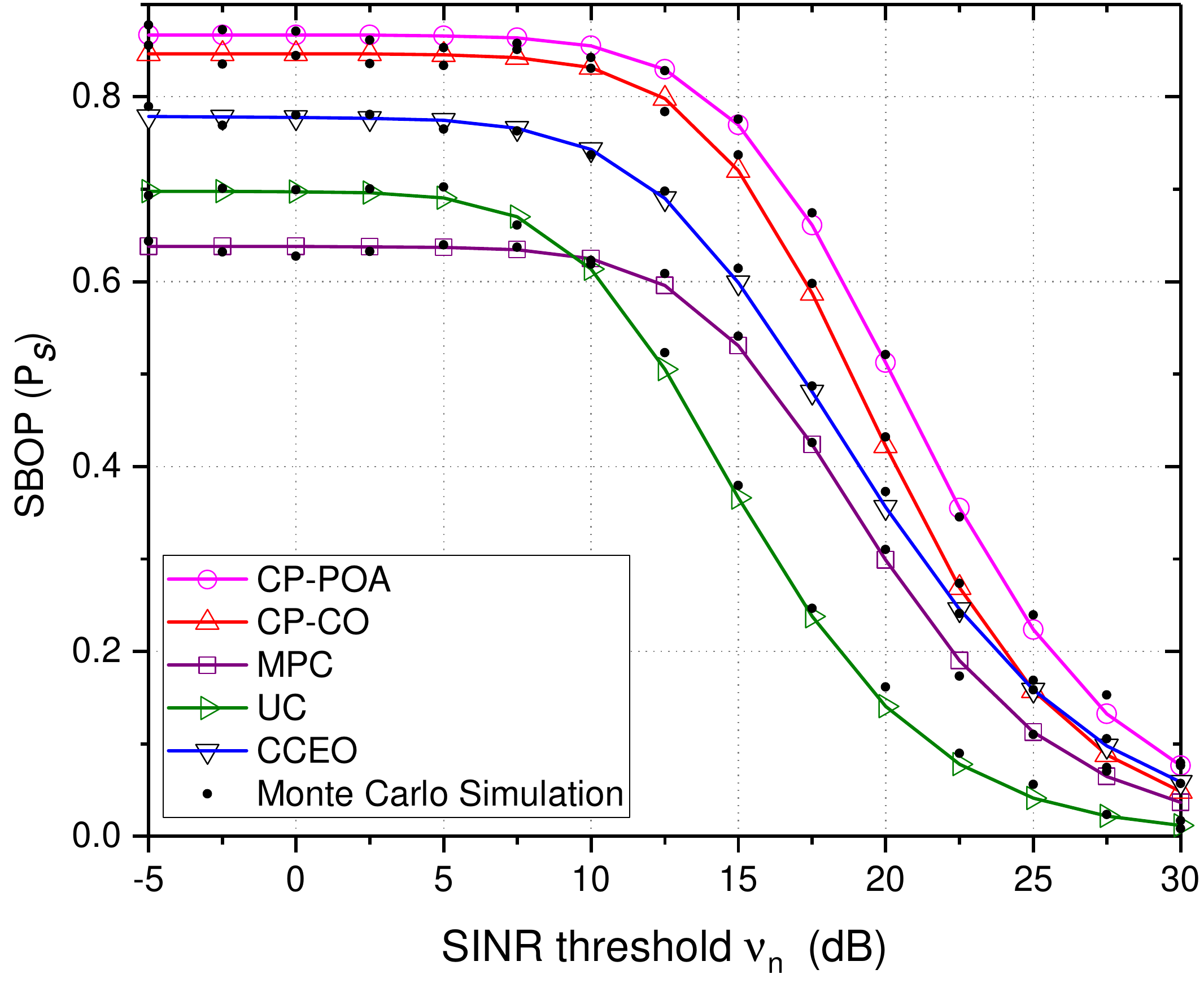}}
    \hspace{0.05in}
  \subfigure[]{
    \label{fig:SBOP-beta} 
    \includegraphics[width=0.3\textwidth]{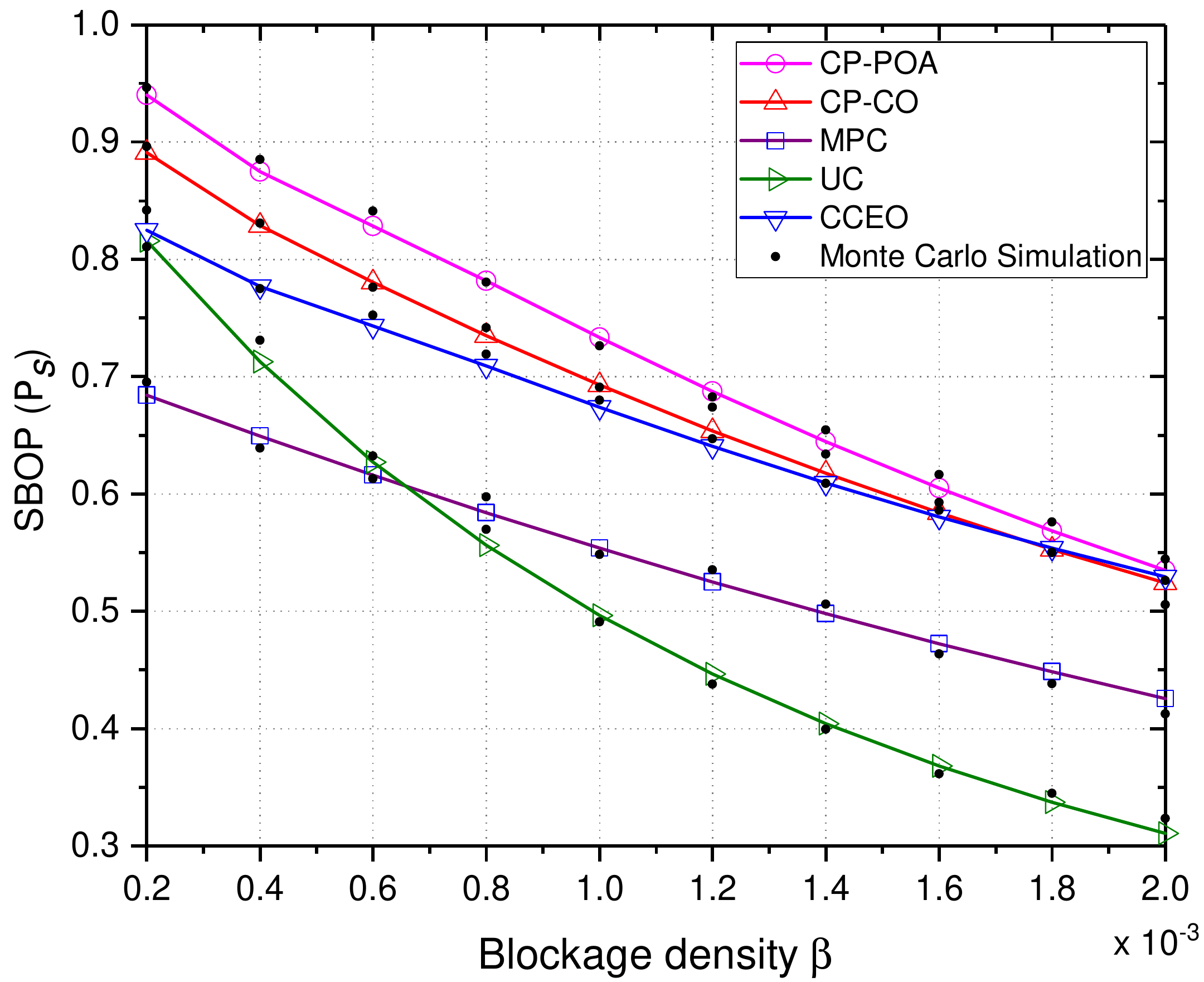}}
  \textcolor{blue}{\caption{Comparison of the performance of SBOP under different caching placement algorithms. }}
  \label{fig:comparison-of-different-strategies} 
\end{figure}

\subsection{Comparison with existing benchmark algorithms}
We then evaluate the performance of SBOP using different caching placement algorithms. It is noticed that the Monte Carlo simulation results closely match the numerical ones.
It is evident from Fig. \ref{fig:SBOP-lambdaRN} that the proposed caching algorithms using CP-POA and CP-CO are both superior to the MPC, UC, respectively, for varying densities of RNs.
It can be observed that the performance of using CP-CO is quite close to that of using CP-POA when the density of RNs is relatively low. This is due to the fact that the interference to a typical UE is smaller when the density of RNs is lower. In this case, adopting CP-CO that ignores the impact of interference will achieve a near-optimal performance compared with the optimal CP-POA algorithm that requires higher computational complexity.
\textcolor{blue}{
It is also observed that adopting CCEO can achieve a near-optimal performance when there is no relay in the network, while the performance improvement is limited when RNs are introduced into the network. This is because the CCEO algorithm does not take into consideration of the spatial correlation caused by the coexistence of BSs and RNs, which makes it hard to obtain the optimal caching placement decision.
}


In Fig. \ref{fig:SBOP-SINRt}, the performance of the SBOP under different algorithms is shown, with the SINR threshold ranging from -5 to 30 dB, which reflects different rate requirements of files. It can be seen that the proposed algorithms outperform the other algorithms for varying SINR thresholds. \textcolor{blue}{Due to the randomly initialized parameters of the CCEO algorithm, it achieves a lower SBOP. }  It can be concluded that a proper caching algorithm is essential for the relay-assisted mmWave network to offload more backhaul traffic, especially for SINR threshold less than 20 dB.

In Fig. \ref{fig:SBOP-beta}, the performance of the SBOP is shown for varying blockage density $\beta$. It can be seen that the SBOP using UC is close to the proposed algorithms in the case of a lower blockage density. This is due to the fact that lower blockages make the channel condition better, which makes the file diversity be dominant to generate the performance gain. Hence, our proposed algorithms will be more inclined to the file diversity, which makes the caching placement result more similar to UC. However, with the increase of $\beta$, the performance of UC declines significantly, while that of the proposed algorithms declines much more steadily. This is because the channel condition becomes worse and more signal blockages will occur, so our proposed algorithms are more inclined to cache popular files to shorten geometric communication distance between the transceivers, which makes links more probability of LOS than NLOS. Therefore, it can be concluded that our proposed algorithms are blockage-aware.


\subsection{Impact of various network parameters on SBOP}
\begin{figure}
  \subfigure[]{
    \label{fig:SBOP-3d} 
    \includegraphics[width=0.36\textwidth]{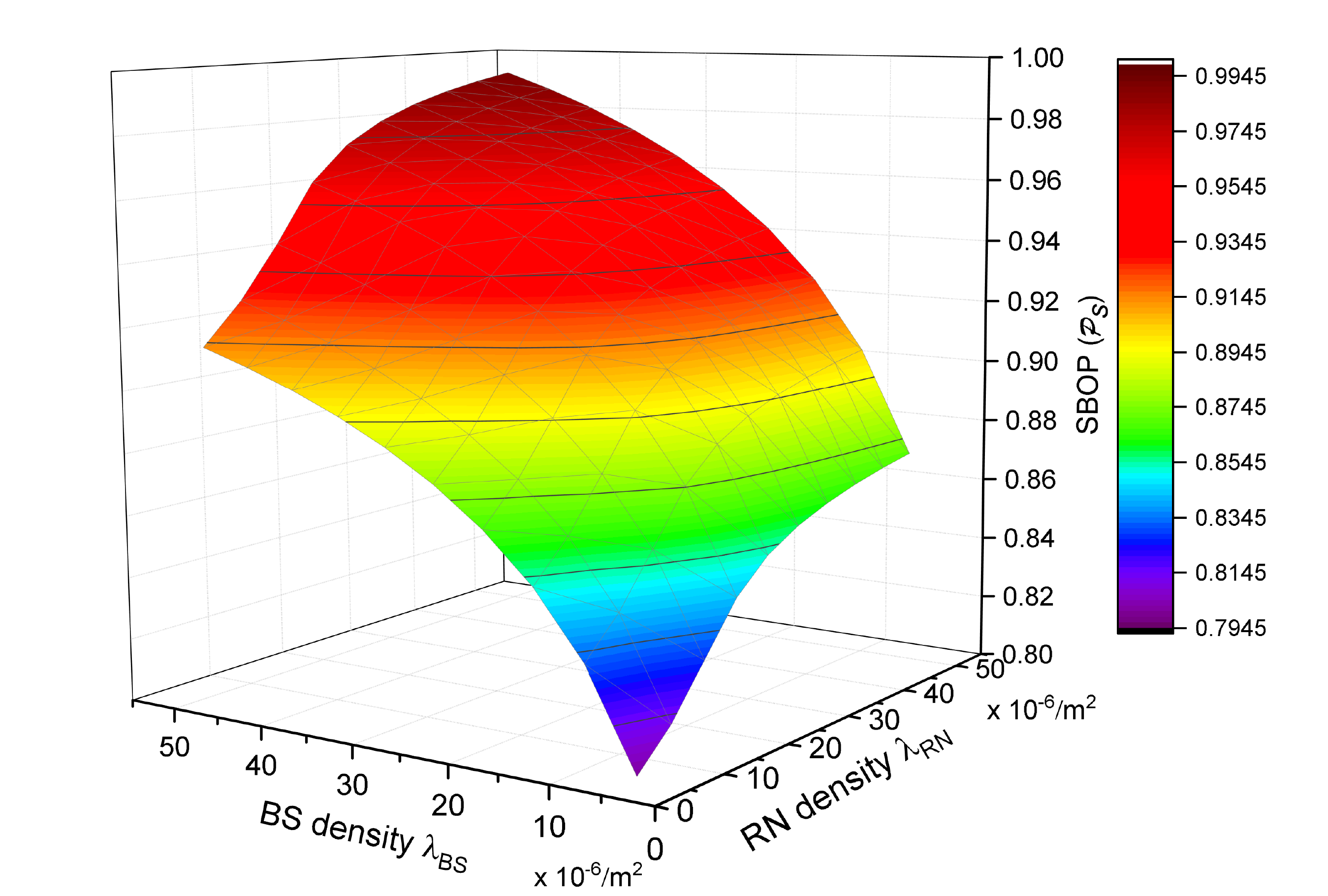}}
  \subfigure[]{
    \label{fig:SBOP-delta-POA} 
    \includegraphics[width=0.29\textwidth]{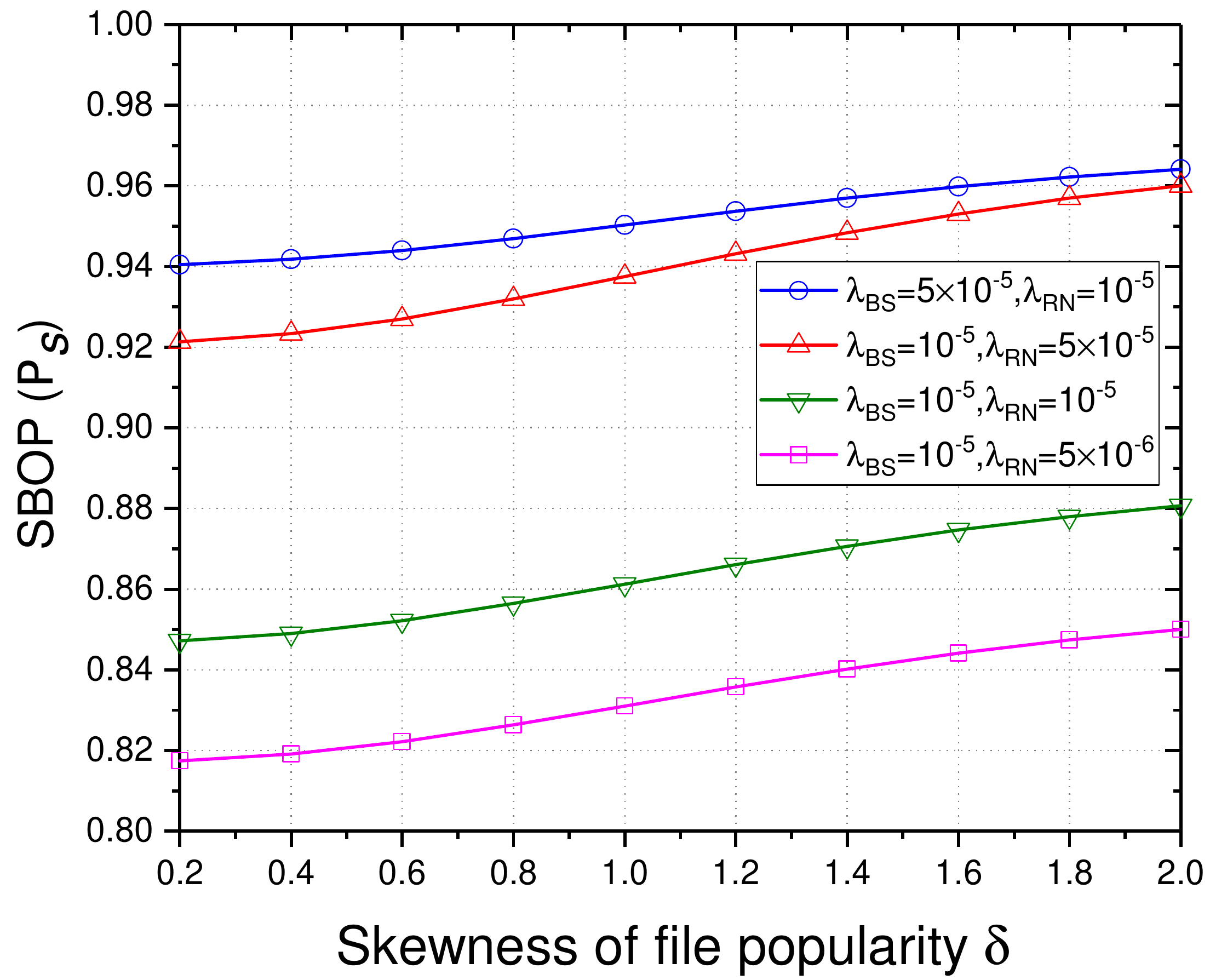}}
  \subfigure[]{
    \label{fig:SBOP-cachesize-POA} 
    \includegraphics[width=0.285\textwidth]{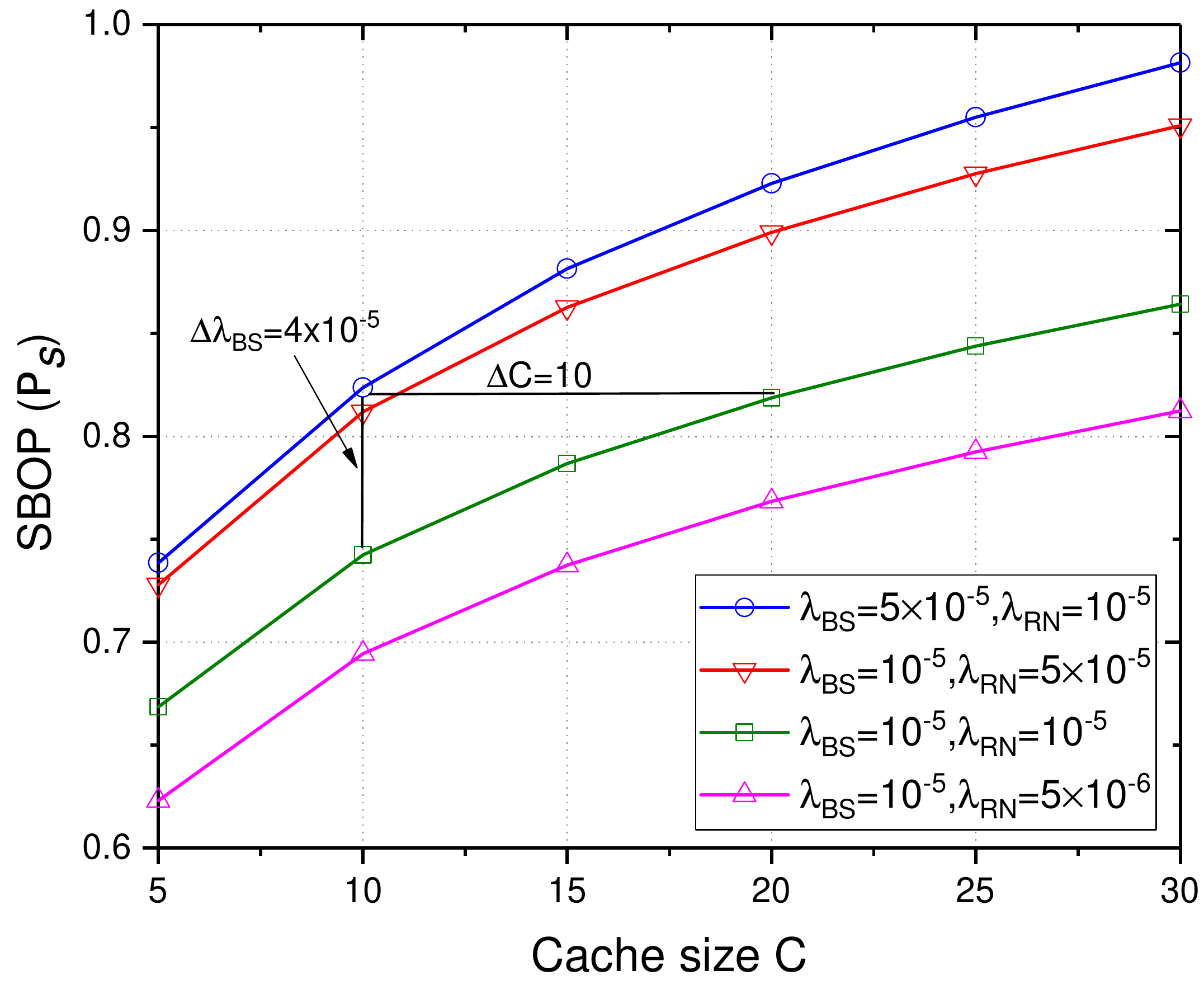}}
  \caption{The performance of SBOP using CP-POA, (a) under various densities of BSs and RNs, (b) under different cache sizes with F = 50, and (c) under different skewnesses of file popularity.  }
  \label{fig:SBOP-POA-performance} 
\end{figure}
To obtain further insights, we evaluate the impact of various network parameters on the SBOP of the proposed CP-POA algorithm. In Fig. \ref{fig:SBOP-3d}, the performance of SBOP for various densities of BSs and RNs is shown. It is observed that SBOP increases with larger densities of BSs or RNs.
Interestingly, the case with $\lambda_{\mathrm{BS}}=$ 10$^{-5}/$m$^2$ and $\lambda_{\mathrm{RN}}=$5$\times$10$^{-5}/$m$^2$ achieves the similar SBOP performance to that with more BS deployed, i.e., $\lambda_{\mathrm{BS}}=$ 5$\times$10$^{-5}/$m$^2$ and $\lambda_{\mathrm{RN}}=$10$^{-5}/$m$^2$, which can be observed in Fig. \ref{fig:SBOP-delta-POA}.
This gives the insights that deploying more relay nodes requiring lower cost than BSs in the mmWave network can also bring comparable system performance, which is a more cost-effective alternative than deploying more mmWave BSs. Fig. \ref{fig:SBOP-cachesize-POA} shows that the SBOP performance increases as the cache size increases. It is observed that instead of increasing the BS density to 5$\times$10$^{-5}/$m$^2$, the system can achieve the same SBOP of 0.78 while keeping the BS density of 10$^{-5}/$m$^2$ just by increasing the cache size from 10 to 20.
These results show great promise of deploying caches in the relay-assisted mmWave cellular networks because it is possible to trade off the relatively cheap storage for reduced expensive infrastructure.

\begin{figure}
  \centering
  \subfigure[CP-POA with RNs]{
    \label{fig:scatter-POA} 
    \includegraphics[width=0.29\textwidth]{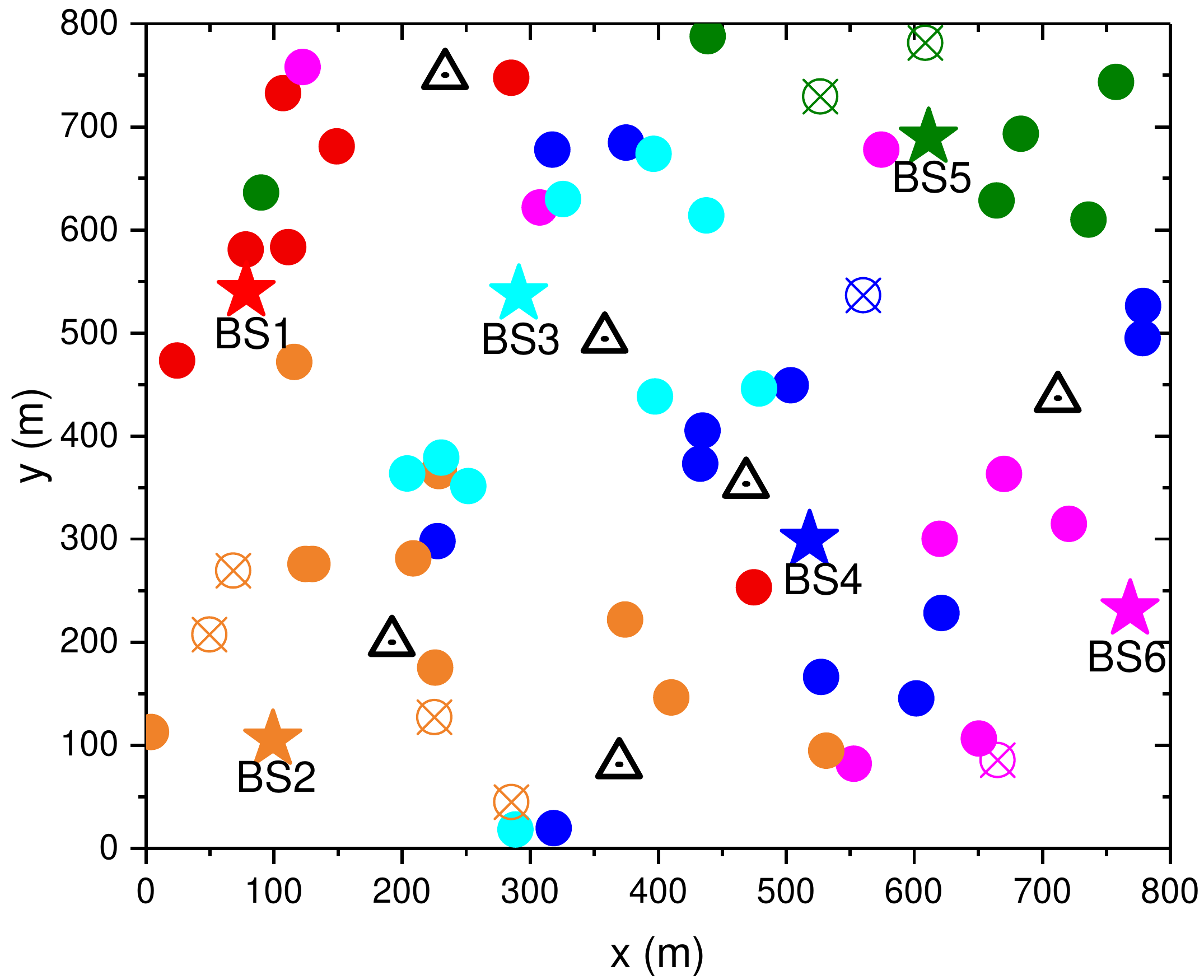}}
  \subfigure[MPC with RNs]{
    \label{fig:scatter-MPC} 
    \includegraphics[width=0.29\textwidth]{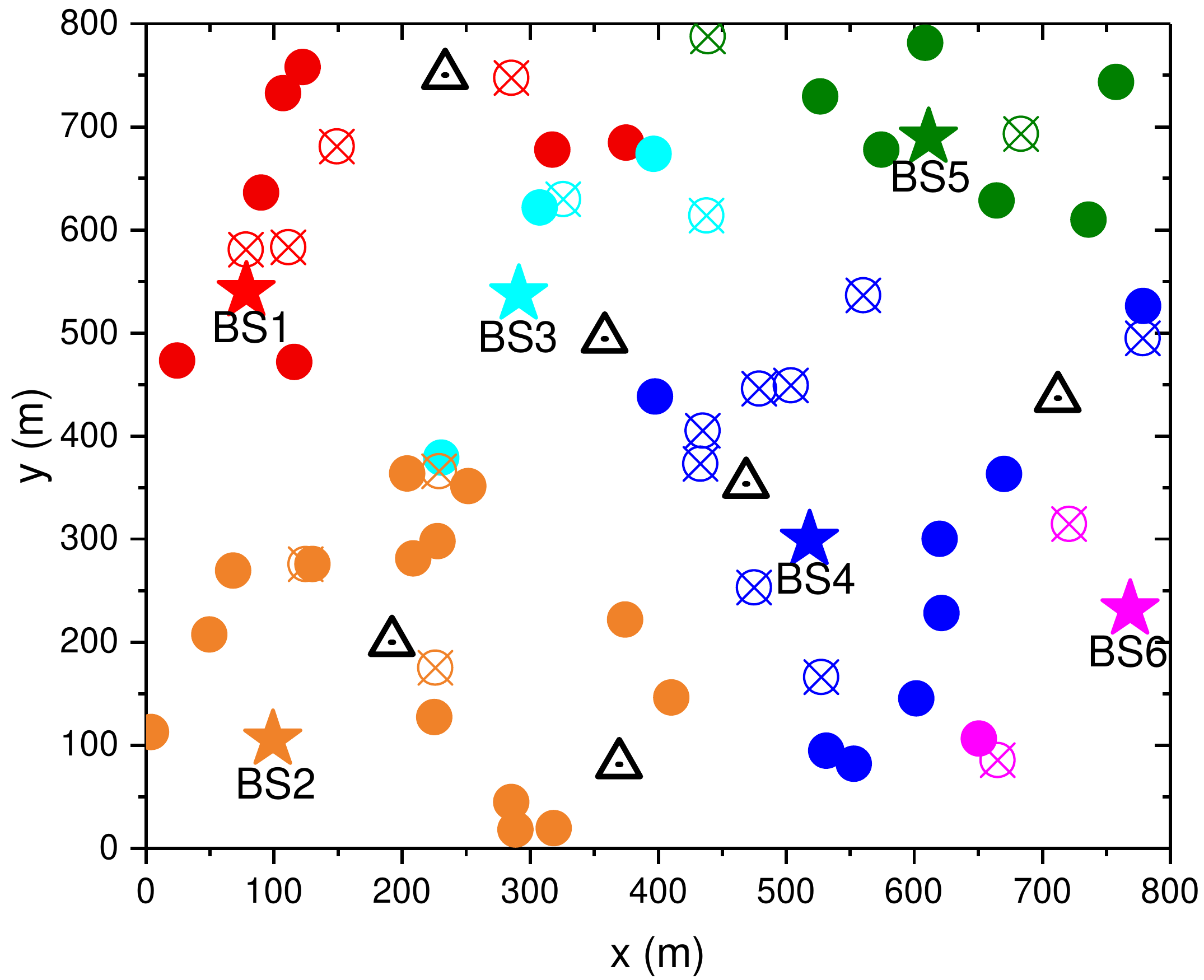}}
  \vspace{3mm}
  \subfigure[UC with RNs]{
    \label{fig:scatter-UC} 
    \includegraphics[width=0.29\textwidth]{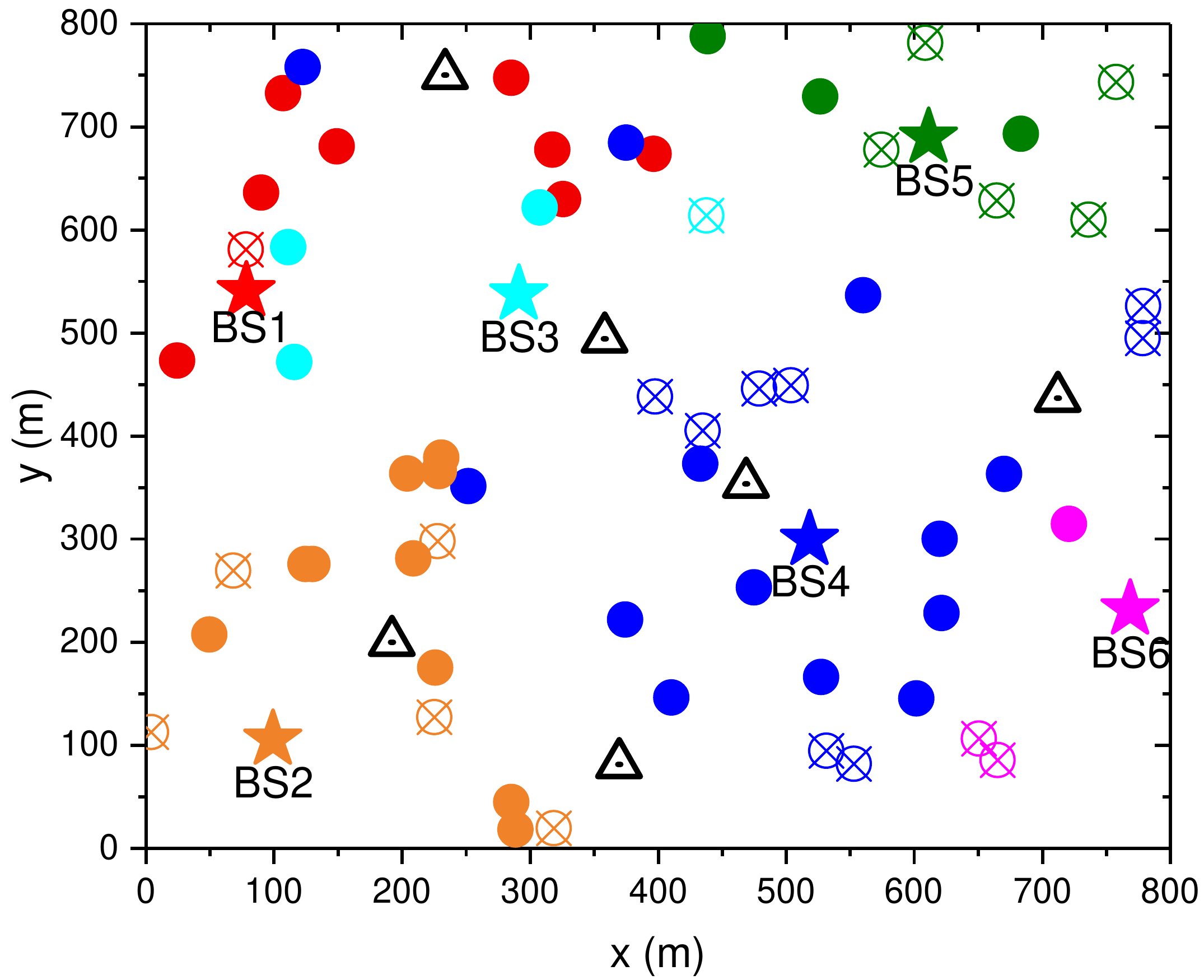}}
  \subfigure[CP-POA without RNs]{
    \label{fig:scatter-POA-norelay} 
    \includegraphics[width=0.295\textwidth]{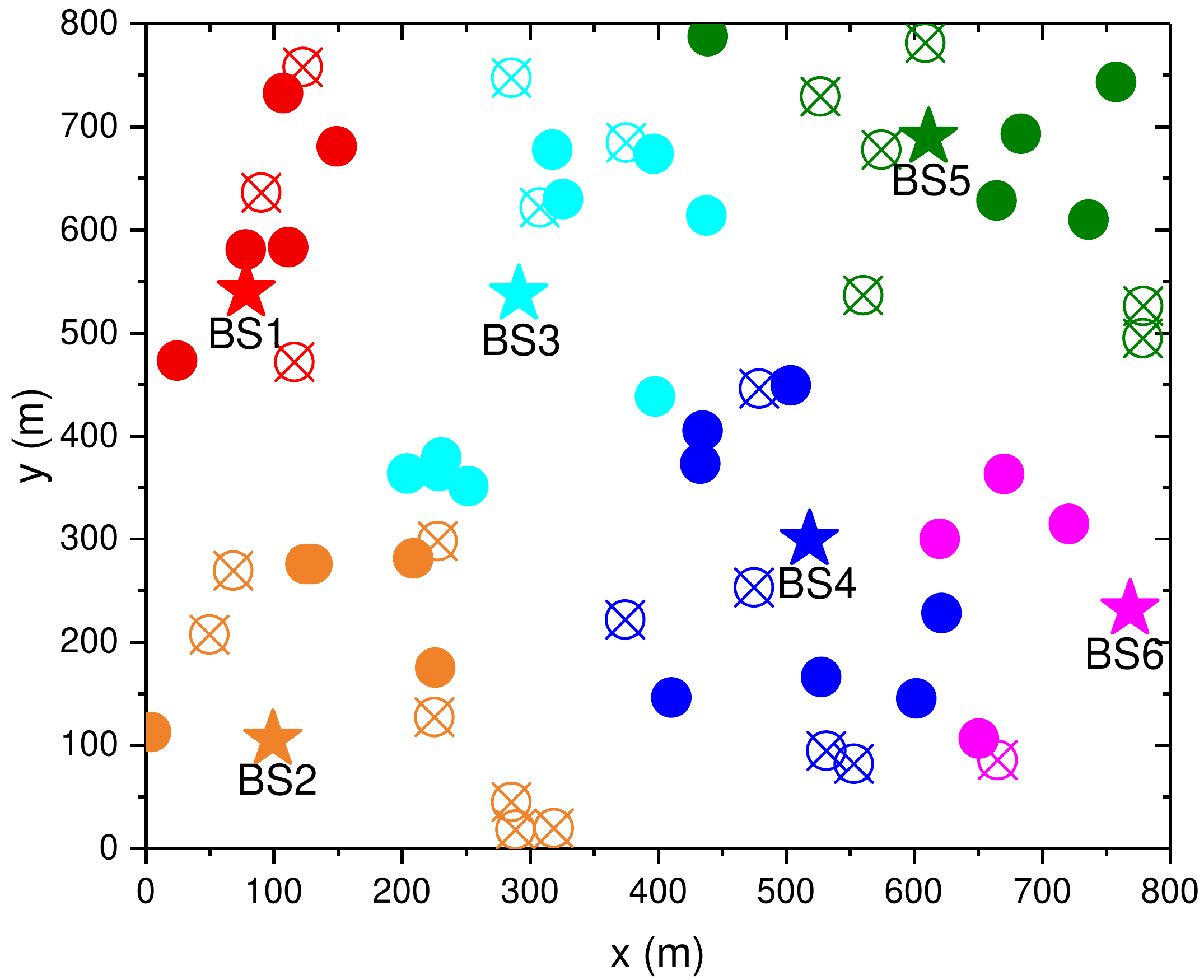}}
    \hspace{-2mm}
  \subfigure[Number of UEs with different \mbox{association types}]{
    \label{fig:numberofUEs} 
    \includegraphics[width=0.289\textwidth]{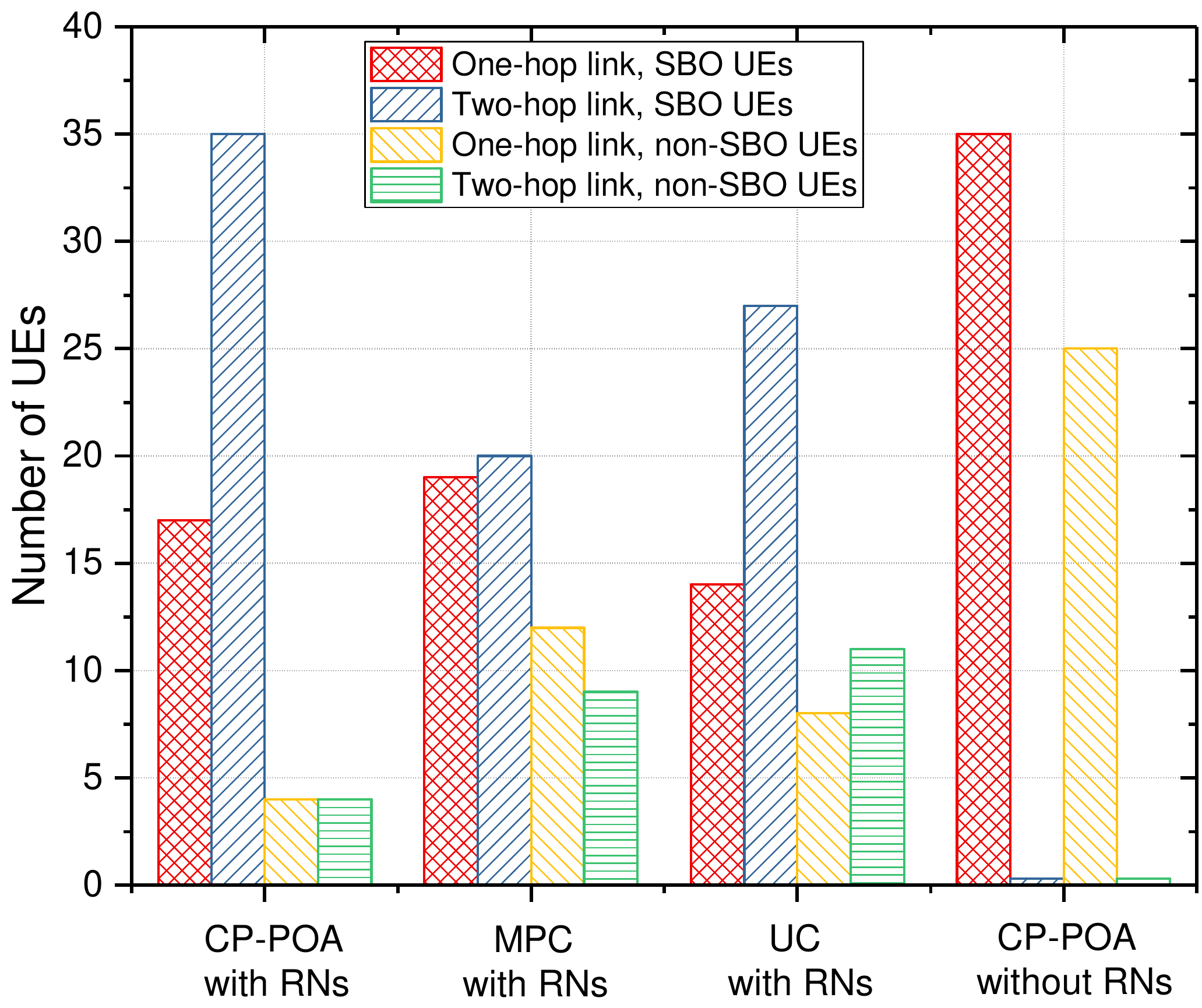}}
    \hspace{-2mm}
  \subfigure{
    \label{fig:legend} 
    \includegraphics[width=0.305\textwidth]{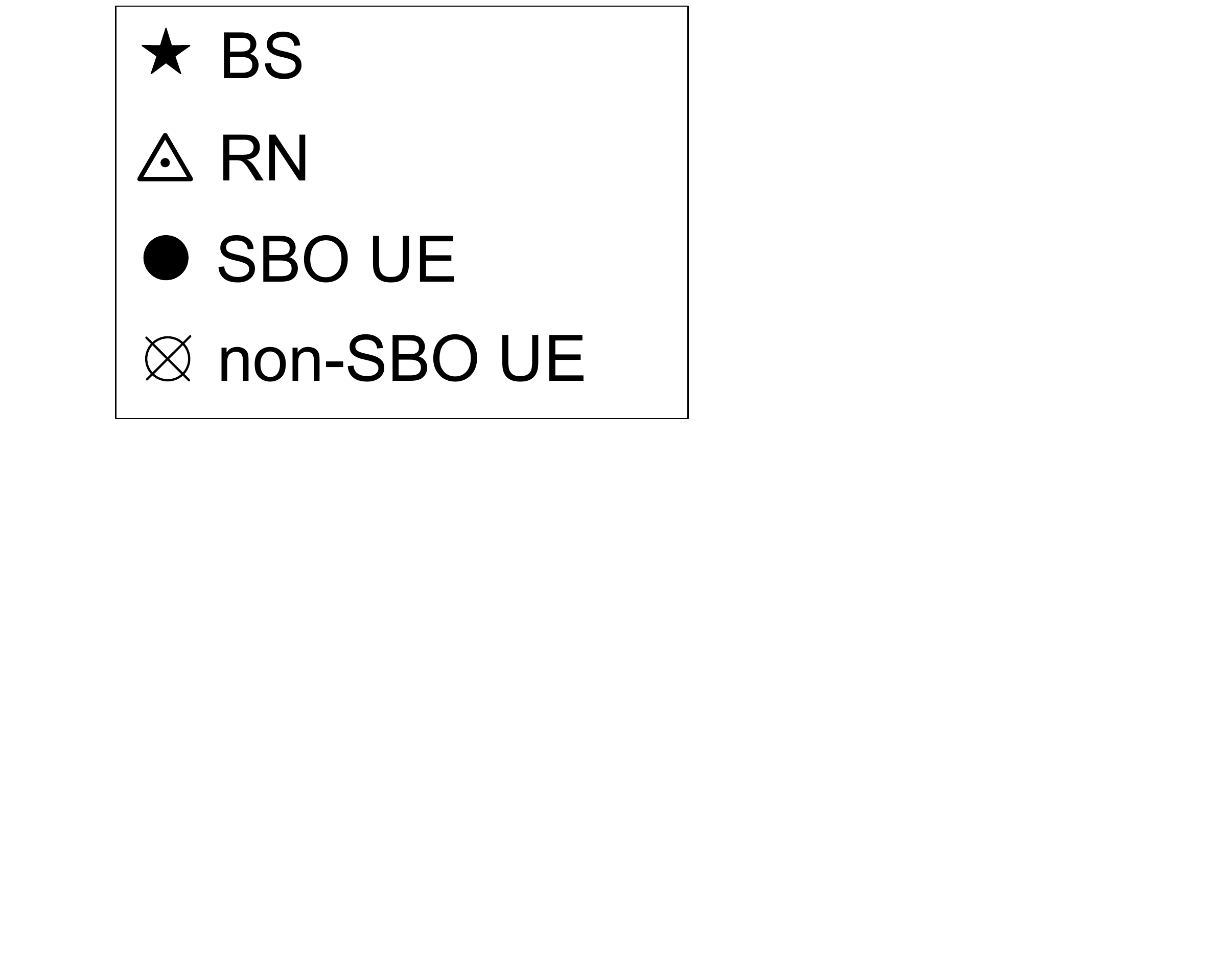}}
  \caption{The impact of RNs and caches on the association of UEs.  }
  \label{fig:scatter-merge} 
\end{figure}

\subsection{Impact of caching and relaying on the association of UEs}
Fig. \ref{fig:scatter-merge} shows the impact of caching and relaying on the association of UEs. A total of 60 UEs are set in the simulation. The same color is used to indicate the association between the BS and the UEs, via either a one- or a two-hop link. UEs who successfully offloaded the backhaul traffic are termed as SBO UEs, otherwise they are termed as non-SBO UEs.
The impact of different caching placement algorithms on the association of UEs is shown in Fig. \ref{fig:scatter-merge}(a), Fig. \ref{fig:scatter-merge}(b) and Fig. \ref{fig:scatter-merge}(c). In general, the number of SBO UEs is the highest using the proposed CP-POA algorithm, as shown in Fig. \ref{fig:scatter-merge}(e), because it takes into account essential factors in the mmWave network such as blockage effects, deployment of RNs, and file diversity gain. It is observed that a UE is most likely to be associated with its nearest BS when using MPC algorithm. This is because all the BSs cache the same popular files, thus reducing the file diversity gain, so that there is no benefit for a UE to associated with a farther BS with the aid of RNs. As a result, the number of SBO UEs via two-hop link decreases from 35 to 20. Therefore with MPC algorithm, introducing RNs will not be fully exploited to improve the performance in the mmWave network. Similarly, as shown in Fig. \ref{fig:scatter-merge}(c) and Fig. \ref{fig:scatter-merge}(e), the number of SBO UEs decreases from 52 to 41 compared with CP-POA. Thus, failing to consider the file popularity when caching files using UC algorithm will also reduce the benefit of introducing RNs.
In addition, Fig. \ref{fig:scatter-merge}(d) shows the association of UEs without RNs, using the proposed CP-POA algorithm. As expected, there are more non-SBO UEs when RNs are not used. For example, this phenomenon is more obvious to the UEs around BS4. Fortunately, with the aid of RNs, blockage effects can be greatly alleviated and the non-SBO UEs around BS4 in Fig. \ref{fig:scatter-merge}(d) can be served by BS1, BS2, and BS3 shown in Fig. \ref{fig:scatter-merge}(a), thus converted to SBO UEs.


\section{Conclusion}
\textcolor{blue}{In this paper, we presented the study on the joint UAR and caching placement optimization in relay-assisted mmWave downlink networks to improve backhaul offloading capability.} We proposed an analytical system model based on stochastic geometry to obtain the insights for the practicality of cache-enabled and relay-assisted mmWave networks.
\textcolor{blue}{To solve the joint UAR and caching placement problem, we first obtained the relationship between UAR probabilities and caching placement probabilities by taking into consideration the spatial correlation caused by the coexistence of BSs and RNs and caching status at BSs. We then transformed the joint optimization problem into a caching placement problem to improve SBOP.} An optimal algorithm with polyblock outer approximation is developed. Furthermore, a suboptimal algorithm based on convex optimization is also designed with low computational complexity. Detailed numerical analysis is performed to validate the effectiveness of the proposed caching algorithms. The results show great potentials of such cost-effective method to deploy caches in relay-assisted mmWave networks to achieve superior backhaul offloading performance.


\appendices
\section{Proof of Proposition \ref{prop-PBS}} \label{AppendixB}
To proof this proposition, we substitute (\ref{SINR-def1}) and (\ref{Shannon-formula}) into (\ref{SCDP-def}). Accordingly, the conditional SBOP by the BSs can be calculated as
\vspace{-3mm}
\begin{align}
\mathcal{P}_{\mathrm{s}}^{(\mathrm{1hop})}(\{\tau_n \}) =  \mathbb{P} \left[\text{SINR}_n^{\mathrm{BU}} > \nu_n \right] \overset{(a)} = \bm{\rho}_\mathrm{L} \mathbb{P} \left[\text{SINR}_n^{\mathrm{BU,L}} > \nu_n \right] + \bm{\rho}_\mathrm{N} \mathbb{P} \left[\text{SINR}_n^{\mathrm{BU,N}} > \nu_n \right], \label{total-probability}
\end{align}
where $\nu_n = 2^{\tau_n / B} - 1$, $\bm{\rho}_\mathrm{L}$ and $\bm{\rho}_\mathrm{N}$ denote the probability that the BU link is in LOS or NLOS state, respectively. $\mathbb{P} \left[\text{SINR}_n^{\mathrm{BU,L}} > \nu_n \right]$ denotes the probability that the SINR of the BU link is above the minimum rate requirement $\nu_n$ when the link is in LOS state. Likewise, $\mathbb{P} \left[\text{SINR}_n^{\mathrm{BU,N}} > \nu_n \right]$ denotes the probability that the SINR of the BU link is above the minimum rate requirement $\nu_n$ when the link is in NLOS state. (a) is obtained by using the law of total probability. Then the first item in (a) can be reduced to

\vspace{-8mm}
\begin{small}
\begin{align}
& \ \bm{\rho}_\mathrm{L} \mathbb{P} \Bigg[\frac{P_{\mathrm{BS}} |h_0|^2 G_{\mathrm{BS}} r^{- \alpha_\mathrm{L}}}{ \underbrace{ \mathcal{I}_{n}^{\mathrm{BU}}\{ \Phi_{\mathrm{BS}}\} + \mathcal{I}_{n}^{\mathrm{BU}}\{ \Phi_{\mathrm{RN}}\} }_{\mathcal{I}_n^{\mathrm{BU}}} + \sigma^2 } > \nu_n \Bigg] \nonumber \\
\overset{(b)}\approx & \ \int_{0}^{\infty} \rho_\mathrm{L}(r) \left\{ 1-\mathbb{E}_{\mathcal{I}_n^{\mathrm{BU}}} \left[ \left( 1-\mathrm{exp} \left( -\frac{\eta_\mathrm{L} \nu_n r^{\alpha_\mathrm{L}}  (\sigma^2 + \mathcal{I}_n^{\mathrm{BU}})}{P_{\mathrm{BS}} G_{\mathrm{BS}}}\right) \right)^{N_\mathrm{L}} \right] \right\} \tilde{f}_{\mathrm{BU}}(r)  \mathrm{d} r  \qquad \nonumber \\
\overset{(c)} = & \ \int_{0}^{\infty} \rho_\mathrm{L}(r) \left\{ \sum_{u=1}^{N_\mathrm{L}} (-1)^{u+1} \binom{N_\mathrm{L}}{u} \mathrm{e}^{ -\frac{u \eta_\mathrm{L} \nu_n r^{\alpha_\mathrm{L}}  \sigma^2 }{P_{\mathrm{BS}} G_{\mathrm{BS}}}} \mathcal{L}_{\mathcal{I}_n^{\mathrm{BU}}} \left( \frac{u \eta_\mathrm{L} \nu_n r^{\alpha_\mathrm{L}} }{P_{\mathrm{BS}} G_{\mathrm{BS}}} \right) \right\} \tilde{f}_{\mathrm{BU}}(r) \mathrm{d} r, \label{BULOS}
\end{align}
\end{small}
$\!\!\!$where $\eta_{\mathrm{L}} = N_\mathrm{L}(N_\mathrm{L}!)^{-\frac{1}{N_\mathrm{L}}}$, $\mathcal{L}_{\mathcal{I}_n^{\mathrm{BU}}}(s)$ is the Laplace transform of $\mathcal{I}_n^{\mathrm{BU}}$ evaluated at $s$. Particularly, \textcolor{blue}{(b) follows from the Alzer's approximation of a gamma random variable \cite{On-some-inequalities}, which is shown to be tight with different system parameters.} Further, (c) follows by using Binomial theorem and the assumption that $N_\mathrm{L}$ is an integer, as well as the Laplace transform of random variable $\mathcal{I}_n^{\mathrm{BU}}$. Next, to obtain the complete expression of the interference point process $\mathcal{I}_n^{\mathrm{BU}}$, we apply the thinning theorem of a PPP by considering blockages and effective antenna gains, then $\mathcal{I}_n^{\mathrm{BU}}$ can be divided into several independent sub-PPPs as shown in the following,
\vspace{-5mm}

\begin{small}
\begin{align}
\mathcal{I}_n^{\mathrm{BU}}  =  \sum_{G} \bigg\{ \mathcal{I}_{n,G}^{\mathrm{BU}}\{ \Phi_{\mathrm{BS}_n,\mathrm{L}}\} \!+ \! \mathcal{I}_{n,G}^{\mathrm{BU}}\{ \Phi_{\mathrm{BS}_n,\mathrm{N}}\} \!+\!  \mathcal{I}_{n,G}^{\mathrm{BU}}\{ \bar{\Phi}_{\mathrm{BS}_n,\mathrm{L}}\} \!+\! \mathcal{I}_{n,G}^{\mathrm{BU}}\{ \bar{\Phi}_{\mathrm{BS}_n,\mathrm{N}}\} \!+\! \mathcal{I}_{n,G}^{\mathrm{BU}}\{ \Phi_{\mathrm{RN,L}}\} \!+\! \mathcal{I}_{n,G}^{\mathrm{BU}}\{ \Phi_{\mathrm{RN,N}}\} \bigg\}, \label{sub-PPP}
\end{align}
\end{small}
$\!\!$where $G \in \{ MM,Mm,mm\}$ denotes the effective antenna gains which is defined in (\ref{Gx-definition}).
Now, the Laplace transform for the interfering links can be expressed as
\vspace{-3mm}
\begin{small}
\begin{align}
& \mathcal{L}_{\mathcal{I}_n^{\mathrm{BU}} } (s_{\mathrm{L}}) = \prod_{i \in \{ \mathrm{L,N}\}} \prod_{G} \Bigg\{ \mathbb{E}_{\mathcal{I}_{n,G}^{\mathrm{BU}}\{ \Phi_{\mathrm{BS}_n,i}\}} \bigg[\mathrm{exp}\Big(-s_{\mathrm{L}}\mathcal{I}_{n,G}^{\mathrm{BU}}\{ \Phi_{\mathrm{BS}_n,i}\}\Big)\bigg]  \nonumber \\
& \qquad \qquad \qquad \mathbb{E}_{\mathcal{I}_{n,G}^{\mathrm{BU}}\{ \bar{\Phi}_{\mathrm{BS}_n,i}\}} \bigg[\mathrm{exp}\Big(-s_{\mathrm{L}}\mathcal{I}_{n,G}^{\mathrm{BU}}\{ \bar{\Phi}_{\mathrm{BS}_n,i}\}\Big)\bigg]  \mathbb{E}_{\mathcal{I}_{n,G}^{\mathrm{BU}}\{ \Phi_{\mathrm{RN},i}\}} \bigg[\mathrm{exp}\Big(-s_{\mathrm{L}}\mathcal{I}_{n,G}^{\mathrm{BU}}\{ \Phi_{\mathrm{RN},i}\}\Big)\bigg] \Bigg\}, \label{interference-expectation}
\end{align}
\end{small}
$\!\!\!$where $s_{\mathrm{L}}=\frac{u \eta_\mathrm{L} \nu_n r^{\alpha_\mathrm{L}}}{P_{\mathrm{BS}} G_{\mathrm{BS}}}$, $\Phi_{\mathrm{BS}_n,i}$ denotes the interfering point process $\Phi_{\mathrm{BS}_n}$ in which the interfering links are in the state of $i \in \{\mathrm{L,N}\}$. Further, (\ref{interference-expectation}) follows from the fact that the sub-PPPs in (\ref{sub-PPP}) are independent. As an example, we compute the expectation of $\mathcal{I}_{n,G}^{\mathrm{BU}}\{ \Phi_{\mathrm{BS}_n,i}\}$ below, and the other terms in (\ref{interference-expectation}) can be derived in a similar method. By utilizing Laplace transform, the above expectation can be calculated as
\begin{small}
\begin{align}
& \ \mathbb{E}_{\mathcal{I}_{n,G}^{\mathrm{BU}}\{ \Phi_{\mathrm{BS}_n,i}\}} \bigg[\mathrm{exp}\Big(-s_{\mathrm{L}}\mathcal{I}_{n,G}^{\mathrm{BU}}\{ \Phi_{\mathrm{BS}_n,i}\}\Big)\bigg] \nonumber \\
\overset{(d)} = & \ \mathbb{E}_{\mathcal{I}_{n,G}^{\mathrm{BU}}\{ \Phi_{\mathrm{BS}_n,i}\}}  \left[ \prod_{{\ell:X_{\ell} \in \Phi_{\mathrm{BS}_n,i} \cap \overline{\mathcal{B}}(0,r)}} \mathbb{E}_g \left[ \mathrm{exp} \bigg( -s_{\mathrm{L}}  g_\ell P_{\mathrm{BS}} G_{\mathrm{BS}} r^{-\alpha_{i}}\bigg)\right] \right] \nonumber \\
\overset{(e)} = & \ \mathrm{exp} \left[ -2 \pi p_n \lambda_{\mathrm{BS}} p_G \int_r^{\infty} \left( 1 - \mathbb{E}_{g} \left[ \mathrm{e}^{-s_{\mathrm{L}} g_\ell P_{\mathrm{BS}} G_{\mathrm{BS}} t^{-\alpha_{i} }}  \right]  \right) \rho_{i}(t) t \mathrm{d} t \right] \nonumber \\
\overset{(f)} = & \ \mathrm{exp} \left[ -2 \pi p_n \lambda_{\mathrm{BS}} p_G \int_r^{\infty} \left( 1 - \frac{1}{\left(1 + s_{\mathrm{L}} P_{\mathrm{BS}} G_{\mathrm{BS}} t^{-\alpha_i} / N_i\right)^{N_i}}\right) \rho_{i}(t) t \mathrm{d} t \right] =  Q_i^{\mathrm{BS}}(p_n, r),
\end{align}
\end{small}
$\!\!$where $\mathcal{B}(0, r)$ denotes the circle centered at the origin of radius $r$, $p_G$ is the probability when the antenna gain takes the corresponding value $G$, $\rho_i(t)$ is the \textit{LOS or NLOS probability function} defined in (\ref{link-state-def}), and $g_\ell = |h_\ell|^2$. In the above, (d) follows from the i.i.d. distribution of $g_\ell$ and its further independence from the point process $\Phi_{\mathrm{BS}_n,i}$, (e) follows by computing the probability generating functional of the PPP, and (f) follows by computing the moment generating function of a Nakagami random variable. \textcolor{blue}{Applying the integral formula of powers of $t$ and powers of binomials, the Laplace transform can be written as the form of Gauss hypergeometric function or Beta function. }

$\!\!$By applying similar methods as mentioned above, the probability that the SINR of the NLOS link for file $n$ is greater than the threshold $\nu_n$ can be calculated as

\vspace{-8mm}
\begin{small}
\begin{align}
\bm{\rho}_{\mathrm{N}}\mathbb{P} \left[\text{SINR}_n^{\mathrm{BU,N}} > \nu_n \right] = \int_{0}^{\infty} \rho_\mathrm{N}(r) \sum_{u=1}^{N_\mathrm{N}} (-1)^{u+1} \binom{N_\mathrm{N}}{u} \mathrm{e}^{ -\frac{u \eta_\mathrm{N} \nu_n r^{\alpha_N}  \sigma^2 }{P_{\mathrm{BS}} G_{\mathrm{BS}}}} \mathcal{L}_{\mathcal{I}_n^{\mathrm{BU}} } (s_{\mathrm{N}}) f_n(r) \mathrm{d} r, \label{BUNLOS}
\end{align}
\end{small}
$\!\!$where $s_{\mathrm{N}} = \frac{u \eta_\mathrm{N} \nu_n r^{\alpha_\mathrm{N}} }{P_{\mathrm{BS}} G_{\mathrm{BS}}}$. \textcolor{blue}{By substituting (\ref{BULOS}), (\ref{interference-expectation}) and (\ref{BUNLOS}) into (\ref{total-probability}), and applying the Gauss-Laguerre Quadrature, the desired proof is obtained.}

\vspace{-1mm}

\section{Proof of Proposition \ref{prop4}} \label{AppendixC}
In the noise-limited scenario, the received SINR is transformed into SNR at the typical UE served by BS$_0$. Thus, the conditional SBOP for the one-hop link can be reduced to
\begin{equation} \small
\mathcal{P}_{\mathrm{s}}^{\mathrm{BS}}(\{\tau_n \}) \approx  \ \mathbb{P} \left[B \log_2 (1 + \text{SNR}_n^{\mathrm{BU}}) > \tau_n \right] =  \mathbb{P} \Bigg[\frac{P_{\mathrm{BS}} g_{n} G_{\mathrm{BS}} r^{- \alpha}}{  \sigma ^2 } >  \underbrace{ 2^{\tau_n / B} - 1 }_{\nu_n} \Bigg]
=  \mathbb{P} \Bigg[\frac{   r^{ \alpha}}{ g_{n}  } \leq  \frac{P_{BS} G_{\mathrm{BS}} }{\sigma ^2 \nu_n}\Bigg]  \label{joint-prob-of-r-g},
\end{equation}
where $\alpha \in \{ \alpha_{\mathrm{L}}, \alpha_{\mathrm{N}} \}$. Firstly, it is useful to calculate the density of $\Psi_{n} = \{ r_{n}^{\alpha} \} \triangleq \{ \psi_{n} \}$.  According to the result in Lemma 1, the intensity measure of this one dimensional PPP is calculated as
\vspace{-4mm}
\begin{align}\label{mapping-thm}
\Psi_n([0,\psi]) = \int_0^{(\psi )^{\frac{1}{\alpha_{\mathrm{L}}}}} 2 \pi p_n \lambda_{\mathrm{BS}} v \mathrm{e}^{-\beta v} \mathrm{d} v + \int_0^{(\psi )^{\frac{1}{\alpha_{\mathrm{N}}}}} 2 \pi p_n  \lambda_{\mathrm{BS}} v (1 - \mathrm{e}^{-\beta v}) \mathrm{d} v,
\end{align}
then the density is calculated as
\begin{small}
\begin{align}\label{density-of-Lambda}
\lambda_n(\psi) = & \ \frac{\mathrm{d} \Psi([0,\psi])}{\mathrm{d} \psi} =   \frac{2 \pi p_n \lambda_{\mathrm{BS}}}{\alpha_{\mathrm{N}}} \psi^{\frac{2}{\alpha_{\mathrm{N}}}-1} + \frac{2\pi p_n \lambda_{\mathrm{BS}}}{\alpha_L} \mathrm{e}^{-\beta \psi^{\frac{1}{\alpha_{\mathrm{L}}}}} \psi^{\frac{2}{\alpha_{\mathrm{L}}}-1} - \frac{2 \pi p_n \lambda_{\mathrm{BS}}}{\alpha_N} \mathrm{e}^{-\beta \psi^{\frac{1}{\alpha_{\mathrm{N}}}}} \psi^{\frac{2}{\alpha_{\mathrm{N}}}-1}.
\end{align}
\end{small}
$\!\!$Now, the density of the process $\Xi_n = \{ \frac{\psi}{g_{n}} (\triangleq \xi_n) \}$ can be obtained by using the displacement theorem \cite[Th. 2.33]{stochastic-geometry-for-wireless-networks}. We first calculate the joint distribution function of $\psi$ and $g_{n}$, which is given by
\vspace{-3mm}
\begin{small}
\begin{align}\label{joint-probability}
& \ \mathbb{P} \left[\frac{   \psi}{ g_{n}  } \leq \xi  \right] = \mathbb{P} \left[ g_{n} \geq \frac{\psi}{\xi} \right] = 1 - F_g(\frac{\psi}{\xi}),
\end{align}
\end{small}
$\!\!$and the joint pdf can be calculated as
\vspace{-3mm}
\begin{small}
\begin{align}\label{joint-PDF}
f(\psi, \xi) = \frac{\mathrm{d} (1 - F_g(\frac{\psi}{\xi}))}{\mathrm{d} \xi} = \frac{\psi}{\xi^2} f_g(\frac{\psi}{\xi}) = \frac{\psi}{\xi^2 \Gamma(m)} m^{m} (\frac{\psi}{\xi})^{m-1} \exp (- \frac{m \psi}{\xi}),
\end{align}
\end{small}
$\!\!\!$where $m \in \{ N_\mathrm{L}, N_\mathrm{N} \}$. Then the density of the process $\Xi_n$ can be calculated by using the displacement theorem, which is given by
\vspace{-3mm}
\begin{small}
\begin{align}
\lambda_{\Xi_n}( \xi) & = \int_0^{\infty} \lambda(\psi) f(\psi, \xi) \mathrm{d} \psi \nonumber \\
& = \int_0^{\infty} \left( \frac{2 \pi p_n \lambda_{\mathrm{BS}}}{\alpha_{\mathrm{N}}} \psi^{\frac{2}{\alpha_{\mathrm{N}}}-1} + \frac{2\pi p_n \lambda_{\mathrm{BS}}}{\alpha_L} \mathrm{e}^{-\beta \psi^{\frac{1}{\alpha_{\mathrm{L}}}}} \psi^{\frac{2}{\alpha_{\mathrm{L}}}-1} - \frac{2 \pi p_n \lambda_{\mathrm{BS}}}{\alpha_N} \mathrm{e}^{-\beta \psi^{\frac{1}{\alpha_{\mathrm{N}}}}} \psi^{\frac{2}{\alpha_{\mathrm{N}}}-1} \right) \nonumber \qquad \qquad \quad \\
& \quad \times \frac{\psi}{\xi^2 \Gamma(m)} m^{m} (\frac{\psi}{\xi})^{m-1} \exp (- \frac{m \psi}{\xi}) \mathrm{d} \psi \nonumber \\
& =  p_n  \pi\lambda_{\mathrm{BS} } \kappa_{\mathrm{N}} \xi^{\kappa_{\mathrm{N}} - 1} \frac{\Gamma(\kappa_{\mathrm{N}} + N_\mathrm{L})}{N_\mathrm{L}^{\kappa_{\mathrm{N}}} \Gamma(N_\mathrm{L})}   +  \sum_{k \in \{ \mathrm{L, N} \}} p_n X(k)   \int_0^{\infty} \psi^{(2/\alpha_k+m-1)} \mathrm{e}^{\big(-\frac{m}{\xi}\psi - \beta \psi^{\frac{1}{\alpha_k}} \big)}   \mathrm{d} \psi,
\end{align}
\end{small}
$\!\!\!$where $\kappa_{\mathrm{N}} = \frac{2}{\alpha_{\mathrm{N}}}$, $X(k)  = \frac{2\pi m^{m} \lambda_{\mathrm{BS}}}{\alpha_k (\xi_0)^{m+1} \Gamma(m)} $.
Now, based on the complementary void function, the cumulative distribution function of $\xi$ is given by
\vspace{-3mm}
\begin{align}
F_{\Xi_n} (\xi_0) = \mathbb{P} [\xi < \xi_0] & = 1 - \mathbb{P} [\Xi_n[0,\xi_0)=0] \overset{(a)} = 1 - \exp \big( - p_n c (\xi_0)^{\kappa_{\mathrm{N}}} - p_n Y(\xi_0) \big)  \label{CDF-of-xi},
\end{align}
where $c = \pi \lambda_{\mathrm{BS}} \kappa_{\mathrm{N}} \frac{\Gamma(\kappa_{\mathrm{N}} + N_\mathrm{L})}{N_\mathrm{L}^{\kappa_{\mathrm{N}}} \Gamma(N_\mathrm{L})}$, and $Y_n(\xi_0) = \sum_{k \in \{ \mathrm{L, N} \}}  X(k)   \int_0^{\infty} \int_0^{\xi_0} \frac{ \psi^{(2/\alpha_k+m-1)}  }{ \xi^{m+1}  } \mathrm{e}^{\big(-\frac{m}{\xi}\psi - \beta \psi^{\frac{1}{\alpha_k}} \big)}  \mathrm{d} \xi \mathrm{d} \psi$. (a) is based on the displacement theorem for point process transformation, which indicates that $\Xi_n$ is also a PPP, and $\mathbb{P} [\Xi_n[0,\xi)=0] = \exp \big(-\int_0^{\xi} \lambda_{\Xi_n}(\xi )\mathrm{d} \xi\big)$. Finally, according to (\ref{joint-prob-of-r-g}) and (\ref{CDF-of-xi}), the conditional SBOP via a one-hop link can be written as
\begin{equation} \small \label{SNR-SCDP}
\mathcal{P}_{\mathrm{s}}^{\mathrm{BS}}(\{\tau_n \})  = \mathbb{P} \Bigg[\frac{   r^{ \alpha}}{ g_{n}  } \leq  \frac{P_{\mathrm{BS}} G_{\mathrm{BS}} }{\sigma ^2 \nu_n}\Bigg]  = F_{\Xi_n} \left(\frac{P_{\mathrm{BS}} G_{\mathrm{BS}} }{\sigma^2 \nu_n}\right)  = 1 - \exp \left( - p_n c \left(\frac{P_{\mathrm{BS}} G_{\mathrm{BS}} }{\sigma^2 \nu_n}\right)^{\kappa_{\mathrm{N}}} - p_n Y_n\left(\frac{P_{\mathrm{BS}} G_{\mathrm{BS}} }{\sigma^2 \nu_n}\right) \right).
\end{equation}
Next, to obtain the conditional SBOP when UE$_0$ is served via a two-hop link in the noise-limited scenario, we can use a similar way, by which Eq. (\ref{prop-equation-PRN}) can also be transformed into a closed-form expression. After some mathematical manipulations, the desired result is obtained.



\end{spacing}
\vspace{7mm}

\bibliographystyle{}

\end{document}